\documentclass[a4paper,11pt]{article}
 \usepackage{biblatex} 
\usepackage{float}
\usepackage{amssymb}
\usepackage{graphicx}
\usepackage{bbm}
\usepackage{mathrsfs}
\usepackage{amsmath}
\usepackage{amsthm}
\usepackage{dsfont}
\usepackage{stmaryrd}
\usepackage[center]{caption}
\usepackage{enumerate}
\usepackage{verbatim}
\usepackage[hmargin=2cm, vmargin=2cm]{geometry}
\usepackage{xcolor}
\usepackage{hyperref}
\usepackage{tikz}

\usepackage{array}

\graphicspath{{figures/}}

\theoremstyle{plain}
\newtheorem{thm}{Theorem}[]
\newtheorem{cor}[thm]{Corollary}
\newtheorem{lem}[thm]{Lemma}
\newtheorem{prop}[thm]{Proposition}

\theoremstyle{definition}

\newtheorem*{rmk}{Remark}

\renewcommand{\P}{\mathbb{P}}
\newcommand{\R}{\mathbb{R}}
\newcommand{\N}{\mathbb{N}}

\newcommand{\ind}[1]{\mathbbm{1}_{\{#1\}}}

\newcommand{\eps}{\varepsilon}
\newcommand{\bp}{\begin{proof}}
	\newcommand{\ep}{\end{proof}}
\def\bal#1\eal{\begin{align*}#1\end{align*}}

\newcommand{\E}[1]{{\mathbb E}\left[#1\right]}

\def\1{\mathds{1}}

\newcommand{\floor}[1]{\left\lfloor #1 \right\rfloor}

\numberwithin{equation}{section}
\numberwithin{thm}{section}

\usepackage{color}

\title{Selection of the fittest or selection of the luckiest: the emergence of Goodhart's law in evolution}

\author{Zsófia Talyigás\thanks{Faculty of Mathematics, University of Vienna, Oskar-Morgenstern-Platz 1, 1090 Wien, Austria} \thanks{Department of Natural Sciences and Sustainable Resources, Institute of Mathematics, BOKU University, Gregor-Mendel-Straße 33, 1180 Wien, Austria}, Francesco Paparella\thanks{Division of Science and Mubadala Arabian Center for Climate and Environmental Science, NYUAD Research Institute, New York University Abu Dhabi, Saadiyat Island, Abu Dhabi, United Arab Emirates}, Bastien Mallein\thanks{Institut de Mathématiques de Toulouse, UMR 5219, Université de Toulouse, UPS, F-31062 Toulouse Cedex 9, France} \thanks{Institut universitaire de France (IUF)},  Emmanuel Schertzer\footnotemark[1]}
\date{\today}

\begin{document}
\maketitle

\begin{abstract}
Natural selection is commonly assumed to become more effective as it becomes stronger. However, selection acts on phenotypes rather than directly on genotypes, and phenotypic success is inherently noisy. Here we study how this mismatch shapes long-term evolutionary dynamics. Using a minimal stochastic model in which individuals inherit genetic fitness while selection acts on noisy phenotypic expressions, we show that increasing selection strength accelerates adaptation only up to a critical threshold. Beyond this point, stronger selection paradoxically slows evolution and erodes genetic diversity by favoring the luckiest individuals rather than the genetically fittest.

We identify two distinct evolutionary regimes—selection of the fittest and selection of the luckiest—separated by a sharp transition. This transition corresponds to a previously unrecognized change in the structure of traveling fitness waves, from semipulled to fully pulled fronts, with profound consequences for adaptation speed and genealogical structure. Our results reveal a biological instance of Goodhart’s law: when phenotypic measures become overly optimized targets, they cease to reliably promote genetic improvement. These findings highlight intrinsic limits to the effectiveness of strong selection and suggest that optimal evolutionary outcomes require intermediate selection strength in noisy environments.
\end{abstract}

\section{Introduction}

Natural selection is often assumed to become more effective as it becomes stronger. Yet selection acts on phenotypes rather than directly on genotypes, and phenotypic success is inherently shaped by chance. As a result, the individuals favored by selection within a given generation are not always those carrying the most heritable advantages. This distinction, already emphasized by Fisher in his dictum that “natural selection is not evolution”\cite{fisher1999genetical}, highlights a fundamental separation between short-term selective success and long-term evolutionary change. Selection filters phenotypes within generations, whereas evolution depends on the reliable transmission of genetic information across generations.

This raises a fundamental question for evolutionary theory: how does the strength of selection shape long-term evolutionary outcomes, not only in terms of the rate of adaptation, but also in terms of the maintenance of genetic diversity?

\medskip

When optimization relies on noisy or imperfect proxies, intensifying selection can produce unintended and even counterproductive outcomes. In economics and the social sciences, this failure mode is known as Goodhart’s law \cite{Goodhart1984,danielsson2002RiskModeling}: when a measure becomes a target, it ceases to be a good measure. 
Closely related phenomena have since been identified across a wide range of disciplines, including statistics, control theory, computer science, and machine learning  where optimizing imperfect objective functions can lead to overfitting and instability \cite{mlGoodhart,reinforcedGoodhart,pan2022misalignment,amodei2016concrete,russell2022human,lehman2020surprising,lehman2011abandoning}. In all these contexts, stronger optimization does not necessarily improve performance, but can instead amplify noise and distort the relationship between measured targets and underlying goals. Whether an analogous phenomenon can arise in biological evolution—where selection measures phenotypic prowess as proxy for underlying genetic fitness—has remained largely unexplored. Yet biological selection is unavoidably noisy, shaped by environmental contingencies, developmental variability, and stochastic interactions between organisms and their environment. These sources of noise suggest that overly strong selection might, paradoxically, undermine the effectiveness of the  evolutionary mechanism.

\medskip

To investigate this possibility, we study an idealized evolutionary model in which individuals inherit a quantitative genetic fitness, while selection operates on noisy phenotypic expressions of that fitness. Only a fraction of individuals reproduce each generation, with reproductive success determined by phenotypes that imperfectly reflect the genotypes transmitted to offspring. This framework captures a minimal but generic feature of biological evolution: selection acts on observable traits, whereas evolution proceeds through the transmission of genetic information. By analyzing the long-term dynamics of this system, we examine how reproduction skewness and phenotypic noise jointly shape evolutionary trajectories.

\medskip

First, we show that strengthening selection by allowing fewer individuals to reproduce—thereby magnifying small phenotypic differences—accelerates adaptation only up to a critical threshold. Beyond this point, stronger selection systematically favors the luckiest individuals rather than the genetically fittest, leading to a slowdown of evolutionary adaptation. This change in behavior reflects a qualitative shift in evolutionary dynamics between two distinct regimes, which we term the selection of the fittest and the selection of the luckiest. In {\it the selection of the fittest} regime, moderate restrictions on reproduction allow selection to reliably preserve the highest-fitness genotypes, resulting in efficient and sustained adaptation. By contrast, in {\it the selection of the luckiest regime}, overly strong restrictions amplify stochastic fluctuations in phenotypic success, causing reproductive success to be dominated by chance and leading to the systematic loss of high-fitness genotypes. We characterize these regimes and the transition between them through a combination of simulations and mathematical analysis.

\medskip

At a mathematical level, this transition is accompanied by a change in the nature of the underlying evolutionary waves \cite{hallatschek2011noisy}. We show that the dynamics shift from {\it semipulled} waves in the selection of the fittest regime to {\it fully pulled} waves in the selection of the luckiest regime, highlighting a previously unrecognized distinction within the class of pulled waves
\cite{kolmogorov1937etude,van2003front}. Beyond its mathematical interest, this distinction has direct biological implications. Our newly identified fully pulled regime represents an extreme form of tip-dominated evolution, in which the contribution of the bulk is negligible. As a consequence, genetic drift at the wave front is expected to be even stronger than in previously studied semipulled waves, leading to an especially pronounced loss of genetic diversity \cite{hallatschek2008gene,hallatschek2010life}.

\medskip

The emergence of Goodhart’s law in evolution reflects a non-linear dependence of the rate of adaptation on selection strength. We uncover a second instance of such non-monotonic evolutionary behavior by examining how genetic diversity responds to phenotypic noise. Reducing phenotypic noise can be viewed as an amplification of the evolutionary signal. We find that genetic diversity depends non-monotonically on noise amplitude: 
in the selection of the luckiest regime, amplification of the evolutionary signal by reducing phenotypic noise decreases genetic diversity, whereas, once the selection of the fittest regime is reached, further noise reduction has the opposite effect.

\medskip

Taken together, these results highlight non-monotonic responses of evolutionary outcomes to selection parameters. Continuous changes in how selection operates—whether by restricting reproduction or by modifying the fidelity of phenotypic information—can lead to qualitatively different evolutionary regimes. Our results suggest that biological evolution is governed by intrinsic non-linearities that limit the effectiveness of strong selection and give rise to complex evolutionary behavior.

\section{The model}
\label{sec:theModel}

We consider an asexual population of $N\gg 1$ individuals, of which $K$ are selected to reproduce and generate $N$ offspring. We will assume that each individual possesses a \emph{genotype}, identified as a real number, representing the genotypic fitness of the individual.
The genotype consists of the inherited information transmitted from parent to child. The expression of that genotype in a given environment is the \emph{phenotype}, which is also represented by a single real number in our idealized model. In Biology, the genotype to phenotype map may be influenced by extrinsic and intrinsic noises such as developmental noise \cite{gavrilets1994quantitative}, phenotypic heterogeneity \cite{bodi2017phenotypic}, cellular noise \cite{hortsch2018characterization}, biological noise \cite{eling2019challenges} and intra--genotypic variability \cite{bruijning2020evolution}. In our model, all of the above effects, plus any environmental contingency in the selection process, are accounted by expressing the number representing the phenotype as the sum of the number representing the genotype plus a zero--mean noise. 
We assume a discrete time dynamics consisting of two sub-steps that repeat for an indefinite number of times.

\begin{description}
\item[Reproduction.] Each of the $K$ selected individuals produces a fixed number of offspring, $r$. These children inherit the genotype of their parent up to an independent random fluctuation owed to the occurrence of random mutations. The phenotype is determined by a random fluctuation of the genotypes.  More precisely, the child of an individual with genotype $g$ has genotype distributed as $g + X$ and phenotype distributed as $g + X + Y$ where $g$ is a real number and $(X,Y)$ is a pair of independent random variables with prescribed densities $f_{X}$ and $f_{Y}$ (to be specified later).
\item[Selection.] Following reproduction, the population consists of $N=r K$ adult individuals. Among them, the $K$ individuals with the largest {\it phenotypes} are selected to survive and give birth to the next generation, transmitting their {\it genotypes}. 
\end{description}

Selection then retains only the highest--valued phenotypes, but only the genotypic information propagates from one generation to the next. In this model,  $(N)$ is interpreted as the number of juveniles at every generation; whereas $(K)$ is the number of those juveniles reaching adulthood and allowed to reproduce.

We parametrize selection using a single quantity, $\gamma$, which measures reproductive skew, or how many
juveniles give rise to adults that successfully reproduce and pass their genes to the next generation. When $\gamma$ is close to $1$, reproductive success is broadly shared among juveniles. When $\gamma$ is small, reproduction is highly skewed, with only a few juveniles reaching adulthood and transmitting their genes. Mathematically, this is defined as
$$
\gamma := \log(K)/\log(N) = \log(K)/\log(r K),
$$ 
so that $K=N^\gamma$ is the number of adults passing their genes to the next generation, and $\gamma$ encodes 
{\it reproduction skewness} in a log scale. This entails that the fertility $r=N^{1-
\gamma}$ of individuals before selection is typically large.
Beyond natural populations, $\gamma$ also has a direct interpretation in artificial selection: it corresponds to an explicit choice of how many individuals are allowed to reproduce and contribute genetically to the next generation. For example, one may think of a Lenski-type long-term evolution experiment~\cite{lenski2017experimental} with the addition of a controlled challenge that affects reproduction: individuals could be required to complete a task (such as swimming through a tube with a nutrient gradient), and only a chosen fraction of the best-performing individuals would be allowed to reproduce. By changing how many individuals are allowed to breed, the experimenter effectively adjusts $\gamma$, the reproduction skewness. In this sense the parameter $\gamma$ acts as an ˝evolutionary knob” that controls how strongly traits influence reproductive success.

Similarly, a breeder may also decide what fraction of the current generation is allowed to reproduce, again corresponding to tuning the parameter $\gamma$. Here we note that we also propose and investigate a version of our model with sexual reproduction at the end of Section~\ref{subsec:logprof} and in the supplementary material (SM) in Section~\ref{sect: sexual}.
\\

\medskip

Finally, we make some assumptions on the genetic noise distribution $f_{X}$, and the phenotype noise distribution $f_{Y}$. For the sake of simplicity, we shall assume that both phenotypic and genotypic noises have tails decaying at least exponentially fast. Let $\alpha\geq1$ and $\mu,\lambda>0$ such that 
$$
f_{X}(x) = C_{\lambda,\alpha} \exp(-(\lambda |x|)^\alpha ), \ \ f_{Y}(x) = C_{\mu,\alpha} \exp(-(\mu |x|)^\alpha ),
$$ 
with $C_{.,.}$ being positive normalization constants. The case $\alpha=1$ corresponds to the Laplace distribution, $\alpha=2$ to the Gaussian distribution.
Up to a change of unit of measure, one can assume without loss of generality that $\lambda = 1$ so that $\mu$ now represents the ratio of the genotypic vs the phenotypic standard deviations. 
The parameter $\mu$ therefore indicates how closely an individual's phenotype reflects its genotype.

\section{Selection of the fittest and selection of the luckiest regimes}\label{sect:phase-transition}
We first discuss the results of numerical simulations. Averaging over many realizations, we find that the evolution of the average genotype $m^{(N)}_n$ of the $n$th generation is well described by the linear law:
$$m^{(N)}_n \approx m^{(N)}_0 + V^{(N)}n.$$ 

The rate $V\equiv V^{(N)}$ of genotypic change depends on the population size $N$: for fixed phenotypic noise $\mu$ and selection pressure $\gamma$, larger populations evolve more rapidly than smaller ones. To allow for a meaningful comparison across different population sizes, we therefore rescale the rate of adaptation and consider the quantity $V/\log(N)$ in the Laplace case ($\alpha=1$) and 
$V/\sqrt{\log(N)}$ in the Gaussian case ($\alpha=2$).

In the model, selection acts explicitly on a single quantitative trait. All other traits, as well as the genomic regions encoding them, are not modeled explicitly and are assumed to be selectively neutral. We nevertheless quantify genetic diversity at a neutral locus that is genetically linked to the selected trait. Although this locus does not directly affect fitness, its diversity is indirectly shaped by selection through linkage, and its genetic variation can be strongly reduced by selection acting on the linked trait.
In a classical way,
neutral genetic diversity is assessed by quantifying the depth of a population's genealogical structure \cite{durrett2008probability} at the locus under selection (or equivalently a locus linked to it).
We keep track of the ancestry of the individuals, and for each run of the model we quantify the number of generations $T_2$ that separates two randomly chosen individuals from their most recent common ancestor. The average age of the most recent common ancestor, $N_e = \mathbb{E}(T_2)$, is an \textit{effective population size}, which measures the level of genetic variability in a population \cite{charlesworth2009effective}. We ran the model both for Laplacian and Gaussian noise distributions, and for population sizes spanning several orders of magnitude. The results are summarized below, and shown in Figure~\ref{fig:2pannelsvne} for the Laplacian case with $N=10^5$. The Gaussian case (with similar qualitative properties) is discussed in SM Section~\ref{subsec:GaussianNumerics}.
\begin{figure}[ht]
\begin{center}
    \begin{tabular}{cc}
       \includegraphics[width=0.49\textwidth]{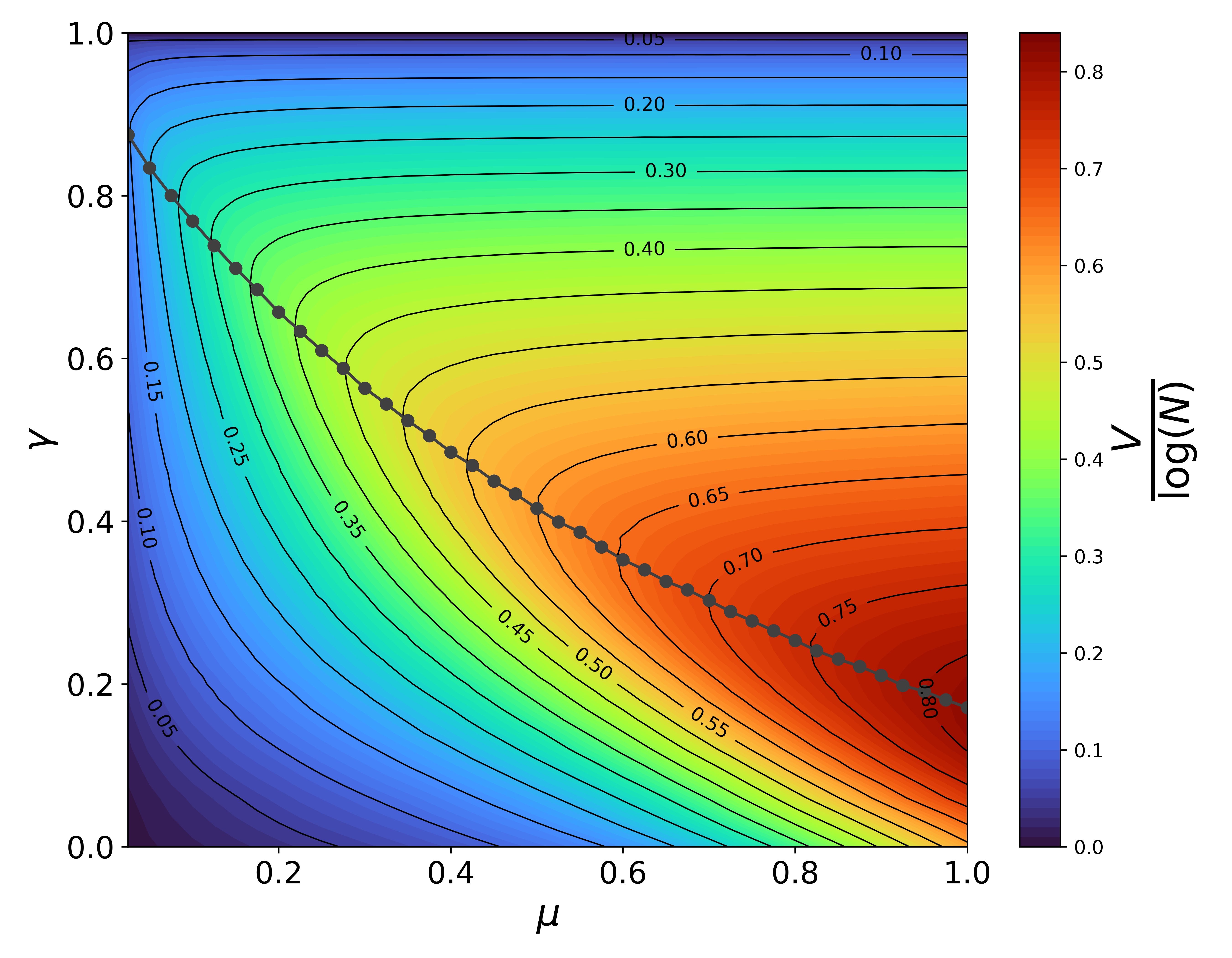}   & 
       \includegraphics[width=0.49\textwidth]{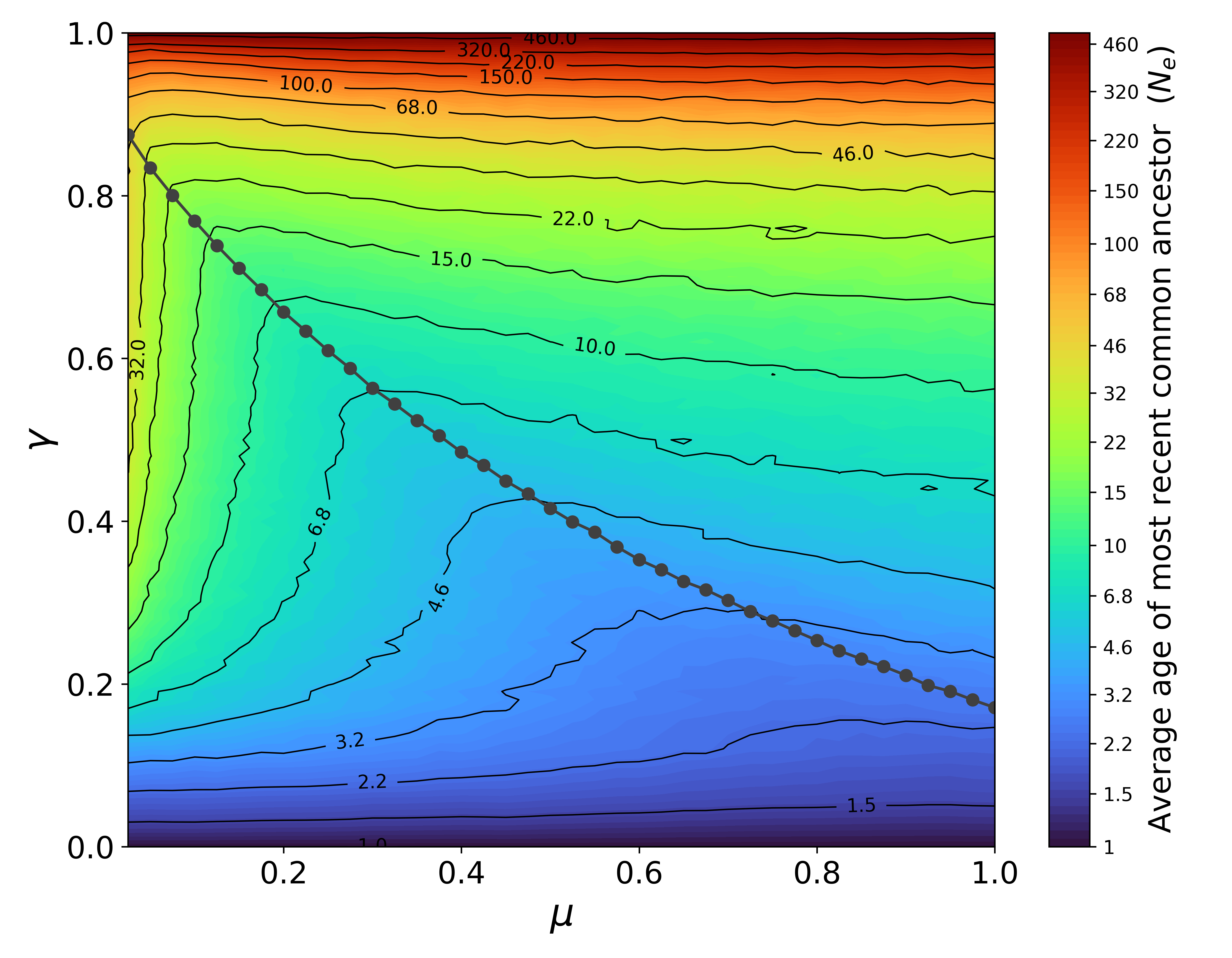} \\
    \end{tabular}
\end{center}
\caption{Rates of adaptation and effective population sizes in our reproduction-selection model with $\alpha=1$ (Laplace distribution for both the phenotypic and genotypic distributions) and a population of size $N=10^5$, plotted as a function of the genotypic to phenotypic standard deviation ratio $\mu$ and the selection pressure parameter $\gamma$ (low $\gamma$ means high selection pressure). 
{\it Left Panel:} Numerically estimated rate of adaptation of the genotype profile. The dotted line is the critical line $\gamma_{c}(\mu)$ defined as the selection pressure that maximizes the rate of adaptation. It separates the regimes of selection of the fittest (lower part of the diagram) and selection of the luckiest (upper part of the diagram). {\it Right Panel:} The effective population size $N_e$, defined as the average age of the most recent common ancestors of random pairs of individuals. The same function $\mu \mapsto \gamma_c(\mu)$ computed from the left panel is reproduced. 
}
\label{fig:2pannelsvne}
\end{figure}

\bigskip

{\bf Rate of adaptation.}
In the left panel of Figure~\ref{fig:2pannelsvne}, we first observe that at a fixed level of noise $(\mu)$, the rate of adaptation is not monotone as a function of the selection parameter ($\gamma$), and we define $\gamma_c(\mu)$ as the value of the selection pressure that maximizes the rate of adaptation for a given value of the phenotypic noise. The function $\mu \mapsto \gamma_c(\mu)$
(black line in Figure~\ref{fig:2pannelsvne})
segregates the parameters $(\mu,\gamma)$ of our models into two domains, that correspond to two regimes for the branching-selection process. 

We refer to the region above the critical line ($\gamma > \gamma_c(\mu)$) as the \emph{selection of the fittest} regime, and to the region below the critical line ($\gamma < \gamma_c(\mu)$) as the \emph{selection of the luckiest} regime. Although this terminology may not be self-explanatory at this stage, its meaning will be clarified in later sections, where we provide an explicit interpretation of the underlying evolutionary dynamics.

In the selection of the fittest regime ($\gamma > \gamma_c(\mu)$), increasing the selection pressure  (i.e.~decreasing $\gamma$ along the vertical axis in the left panel of Figure~\ref{fig:2pannelsvne}) has the effect of increasing the rate of adaptation. 
In the selection of the luckiest regime ($\gamma < \gamma_c(\mu)$), this effect is counter-intuitively reversed: increasing the selection pressure has the effect of \emph{decreasing} the rate of adaptation of the population.
 
As a result, for a fixed value of the noise $(\mu)$, a selection pressure $\gamma_c(\mu)$ entails a maximal rate of adaptation for the population. The selection pressure is too weak ($\gamma>\gamma_c(\mu)$) in the selection of the fittest regime to optimize the speed of adaptation, whereas it is too strong ($\gamma<\gamma_c(\mu)$) in the selection of the luckiest regime and has a detrimental effect on the evolution of the population. This latter case corresponds to an instance of Goodhart's law, i.e., an overly zealous attempt at maximizing a proxy measure (the phenotypic values) hinders the intended goal of maximizing the target (genotypic values). 

\medskip

\noindent
{\bf Effective population size.} 
A fundamental question in population genetics is how natural selection shapes neutral genetic diversity \cite{fisher1999genetical}. Since the early developments of the field \cite{fisher1923xxi,haldane1927mathematical,haldane1963polymorphism,smith1974hitch}, it has been widely accepted that stronger directional selection reduces genetic diversity. Our model is consistent with this classical view, as it predicts an effective population size $N_e \ll N$. 

However, the model reveals a non-intuitive dependence of neutral genetic diversity, as quantified by $N_e$, on the level of phenotypic noise ($\mu$). Decreasing noise (increasing ($\mu$), with ($\gamma$) remaining constant) increases the instantaneous efficiency of selection by tightening the correspondence between genotype and phenotype. Thus, one may be tempted to assimilate a decrease in noise to a decrease in the effective population size $N_e$. In contrast, when $(\gamma)$ is kept fixed,  
we find that genetic diversity depends non-monotonically on the noise amplitude $(\mu)$, as shown in the right panel of Figure~\ref{fig:2pannelsvne}. Further,  the neutral genetic diversity is minimized at a noise level close to the critical line separating the fittest from the luckiest regime.

\medskip

In conclusion, our numerical simulations suggest the existence of a critical line in the selection parameters $(\mu,\gamma)$ where a transition occurs in the rate of adaptation and the ancestral structure of the population. The next sections will be dedicated to provide a theoretical explanation behind this phenomenon and provide a clear explanation for our terminology (luckiest vs fittest).

\bigskip

\section{Deterministic evolution and ancestral structures}
\label{subsec:model}
In order to explain the existence of a critical line where the key evolutionary parameters of the model (rate of adaptation and effective population size) exhibit a non-monotone response to the selection parameters (selection pressure ($\gamma$) and phenotypic noise ($\mu$)), we turn to a mathematical analysis of the model. Here we shall focus on the Laplace case ($\alpha=1$) which is particularly amenable to analysis. 

\subsection{Log-profiles.}\label{subsec:logprof}
Numerical simulations (see top left panel of Fig. \ref{fig:4pannelsprofiles}) reveal that genotypes and phenotypes are typically highly concentrated around the mean of the population. Moreover, the distribution of genotypes in the population, on a logarithmic scale, appears to quickly stabilize on a profile, that travels as a wave over time. We take interest in the dynamic of this logarithmic profile, which describes the extreme genotypes in the population. The typical highest increment of the genotype in a generation being $O(\log(N))$, we rescale this profile appropriately.

More precisely, we define the genotypic profile $g$ of a population as a quantity valued in $\R_+ \cup\{-\infty\}$ so that the number of particles (genotypes) around $x\log(N)$ is approximately given by $N^{g(x)}dx$. In other words, $g(x)$ can be thought of as the limiting stochastic exponent of the population in $N$. In PDE, this is often referred to as the Hopf-Cole transformation of the system (see e.g.~\cite{champagnat2023filling}). 
Note that when $g(x)=-\infty$, this corresponds to having no particle present around $x$. We also consider the phenotypic profile $p(x)$ so that $N^{p(x)}$ captures the number of phenotypes
{\it to the right} of $x\log(N)$ after the reproduction step.

Let $(X,Y)$ be a pair of independent random variables with Laplace distributions of parameters $1$ and $\mu$ respectively and which respectively correspond to the genotypic and phenotypic noise. Direct computations (see SM Section~\ref{sect:tails}) yield the following asymptotic relationships, 
\begin{equation}
\begin{split}
\lim_{N \to \infty} \frac{1}{\log N}  \log \P(X + Y >  y \log N) &= -\min(\mu,1) y_+ \\
  \text{and}\quad \lim_{\epsilon \to 0} \lim_{N \to \infty} \frac{1}{\log N}  \log \P(|X - x \log N|<\eps\log N, X+Y > y \log N) &=-(|x| + \mu(y-x)_+ ),
\end{split}
\label{eq:tail}
\end{equation}
where we write  $x_+ = \max (x,0)$. 
Writing $g_{n-1}$ for the genotypic profile of a population at the $(n-1)$th generation, and $p_n$ for the phenotypic profile of its children,
we observe that $(g_n,p_n)$ evolve according to the following deterministic dynamics:
\begin{equation}
  \label{eq: deterministicDynamic}
  \begin{split}
    p_{n} (x) &= \pi\left[1-\gamma + \sup_{y \in \R}\left( g_{n-1}(y) - \min(1,\mu) (x-y)_+\right)\right]\\
    s_{n} &= \sup\{x \in \R: p_{n}(x) \geq \gamma \}\\
    g_{n}(x) &= \pi\left[1 - \gamma + \sup_{y \in \R} \left( g_{n-1}(y) - |x-y| - \mu (s_n - x)_+\right)\right],
  \end{split}
\end{equation}
with 
\begin{eqnarray*}
\pi(x)= \left\{ \begin{array}{cc} x & \mbox{if $x \geq 0$} \\ - \infty & \mbox{if $x < 0$}. \end{array}\right.
\end{eqnarray*} Let us provide a quick justification for these formulas based on the tail estimates (\ref{eq:tail}). 
The $N^{g_{n-1}(y)}dy$ individuals with a genotype approximately $y \log N$ at generation $n-1$ will create on average 
$$N^{g_{n-1}(y) +1-\gamma - \min(1,\mu) (x-y)_+}dy$$ 
phenotypes larger than $x \log N$ at generation $n$, where we used the first equation of~(\ref{eq:tail}). As a consequence, the number 
of phenotypes larger than $x\log N$ is given by the above formula after integrating over $y$ and using the Laplace method --which can be detected in the $\sup$ of the first equation in~\eqref{eq: deterministicDynamic}. The projector $\pi$ expresses that if the average number of phenotypes larger than $x \log N$ goes to $0$ as $N \to \infty$ then the probability of observing such an individual also becomes vanishingly small. The properties discussed in the rest of the article are derived from studying the deterministic dynamical system~\eqref{eq: deterministicDynamic}. These results translate to the stochastic system in the sense that they hold with high probability over finite time horizons.

The value $s_n$ is then defined such that there are about $N^\gamma$ individuals with phenotypes larger than $s_n \log N$. 
Finally, we obtain $g_n(x)$ by estimating (in the same manner) the number of children with a genotype around $x \log N$ and a phenotype larger than $s_n \log N$, applying the second equation of~(\ref{eq:tail}). The function $\pi$ has the effect of only conserving the positive part of $g_n$, since if $g_{n-1}(y) - |x-y|-\mu(s_n-x)_+ +1-\gamma < 0$, then with high probability none of the individuals at position $y \log N$ will have a descendant at position $x \log N$ with a phenotype larger than $s_n \log N$. We provide a more detailed explanation in SM (Section~\ref{sect:RecursiveExpr}) and a fully rigorous argument for a similar model is given in~\cite{desmarais2025k}.

We show in SM (Proposition~\ref{prop: sngn}) that the evolution can be rephrased in terms of a discrete ``free-boundary" problem
\begin{equation}
  \label{eq: deterministicDynamic2}
  \begin{split}
      g_{n}(x) \ = \  \pi\left[1 - \gamma + \sup_{y \in \R} \left( g_{n-1}(y) - |x-y| - \mu (s_n - x)_+\right)\right], \\
      \mbox{where $s_n$ satisfies $\sup g_{n}= \gamma$},
  \end{split}
\end{equation}
where the second condition reflects the fact that the phenotypic threshold $s_n$ can  be determined solely by the condition that there are $N^\gamma$ individuals left after enforcing the selection step. 
Our computations have been performed for $\alpha=1$. Similar computations for any $\alpha>1$ yield a recursive equation for $g_n$ in the spatial scaling $\log(N)^{1/\alpha}$, which corresponds to the scaling of the fitness wave's diameter.

In Section~\ref{sect: sexual} we also present a sexual variant of our model, which leads to a modified version of the recursive free boundary problem~\eqref{eq: deterministicDynamic2}:
\begin{equation}
\label{eq: deterministicDynamic2sex}
\begin{split}
g_{n}(x) \ = \  \pi\left[ \sup_{y \in \R} \left( 1 - 2\gamma + 2 g_{n-1}(y) - |x-y| - \mu (s_n - x)_+\right)\right], \\
\mbox{where $s_n$ satifies $\sup g_{n}= \gamma$}.
\end{split}
\end{equation}
In Section~\ref{sect: sexual} we derive this formula and highlight the similarities and main differences between the behaviors of the asexual and sexual models.

\medskip

\subsection{Traveling wave solution}\label{sect:traveling-wave}

We only take interest in the long-term asymptotic behaviour of the profile $(g_n)$ defined by \eqref{eq: deterministicDynamic}. 
We say that a function $g$ is a traveling wave for the dynamic \eqref{eq: deterministicDynamic} with speed $v$ if, assuming that $g_0 = g$, we have
\begin{equation}
  \label{eqn:travelingWave}
  g_n(x) =  g(x-nv).
\end{equation}
In other words, the dynamic has the effect of shifting the genotypic profile by $v$, where $v$ is interpreted as the speed of evolution in the natural scale of the system ($\log(N)$ for $\alpha=1$, as in the numerical simulations of Fig.~\ref{fig:2pannelsvne}; and $\sqrt{\log(N)}$ for $\alpha=2$ as in Fig.~\ref{fig:2pannelsvne_Gaussian} in SM).
Examples of such traveling wave solutions, and convergence to those, are depicted in Fig.~\ref{fig:4pannelsprofiles} (bottom and top right).

In SM, we show the existence and uniqueness of a traveling wave solution under minimal assumptions (Theorem~\ref{thm:main}). The crucial part of this result is the existence of a transition segregating the parameter space $(\mu, \gamma)$ into two sub-regions delimited by an explicit curve
\begin{equation}\label{eq:gamma_c}
  \gamma_{c} : \mu \mapsto \gamma_c(\mu) := \frac{\floor{1/\mu}}{\floor{1/\mu}+1} \left( 1 - \frac{\mu}{2 - \floor{1/\mu}\mu}\right),
\end{equation}
drawn in the right panel of Fig.~\ref{fig: rateOfCvSM} in SM (in solid black), and corresponding to distinct evolutionary regimes mirroring our numerical simulations. 
Our deterministic analysis allows us to characterize the two regimes in different ways, which will reflect the numerical observations of Fig.~\ref{fig:2pannelsvne}. This is explained in detail in the next sections. 

\bigskip

\subsection{Selection of the fittest or selection of the luckiest.} First, our mathematical results (Theorem~\ref{thm:main} in SM) identify  the phase transition observed in numerical simulations for the rate of adaptation in the left panel of Fig.~\ref{fig:2pannelsvne}. Below the critical curve ($\gamma<\gamma_{c}(\mu)$), lowering selection by increasing $\gamma$ has the effect of increasing  the rate of adaptation. 
Above the critical curve ($\gamma>\gamma_{c}(\mu)$) the effect is reversed so that the optimal level of selection is attained at the intermediate level $\gamma_{c}(\mu)$. 

Our analysis reveals a clear explanation behind this phenomenon. Consider the population right after reproduction (that is before implementing the selection step). There are $N=rK$ genotypes available among which $K<N$  will be chosen according to their phenotypes in the selection phase. We now ask the following question: since selection only acts on phenotypes, do we always pick the very best genotype?  The answer depends on the evolutionary regime at hand. 
\begin{itemize}
\item {\it  $\gamma<\gamma_{c}(\mu)$: the best genotype is never picked (selection of the luckiest);}
\item {\it $\gamma>\gamma_{c}(\mu)$: the best genotype is always selected (selection  of the fittest)}.
\end{itemize}
To give an intuition behind this phenomenon, we note that when selection is strong ($\gamma<\gamma_c$), we only allow a tiny fraction of the population to reproduce. This is obviously a risky strategy since selection only picks a few individuals whose phenotype can potentially inflate their underlying genotypes. In contrast, weak selection ($\gamma>\gamma_c$) enforces a diversification of the risk so that the highest genotype is always picked. In other words, the extremes of the phenotypic and genotypic spaces are partially decorrelated, and when selection is too strong, it will miss the children with exceptionally high genotypes.

Beyond the intuition, the previous statement is made mathematically precise in Theorem~\ref{thm:reprod} and the preceding paragraphs in SM. The content of this formal result is related to geometric properties of the traveling wave solution, as illustrated in the bottom panels of Figure~\ref{fig:4pannelsprofiles}. Essentially,
the distinction between the selection of the luckiest and fittest regimes can be seen in the position of the phenotypic threshold $s$. Recall that $s$ corresponds to the minimal phenotypic value in order to be selected at the next generation. 
We assume that the population has reached the traveling wave state and we start from a population with a genotypic makeup corresponding to the traveling wave (blue) at time $t$, and make one step of the evolution to obtain the profile at time $t+1$ from~\eqref{eq: deterministicDynamic} (green). This genotypic profile is constructed in two successive steps: we first generate the genotypic profile of the children (orange), and then obtain the genotypic profile at time $t+1$ (green) by thinning the orange profile with our selection procedure. In particular, the difference between the orange and the green curves corresponds to the log-number of genotypes eliminated during the selection process. We can now distinguish between two geometries of the wave.

In the selection of the luckiest regime (Fig.~\ref{fig:4pannelsprofiles} bottom left), the threshold $s$ exceeds the genotypic profile of the children (orange), indicating that in surviving individuals, noise must necessarily provide an exceptionally large boost to the phenotype in order to reach the high phenotypic threshold $s$.
As it is apparent from the figure, the tip of the  reproduction curve (orange) is strictly higher than the genotypic curve at time $t+1$ (green) indicating that all the children with the best genotypes are washed out by selection. (All the children located between the two tips are not selected.)  
Therefore, in the selection of the luckiest regime, the survival of each individual is partially explained by having an unusually high phenotype. The (few) individuals with a very high genotype will tend to have average phenotype, those are therefore not preserved by the selection step, which explains the drop in the rate of adaptation of the population.

In contrast, in the selection of the fittest regime (Fig.~\ref{fig:4pannelsprofiles} bottom right), all the children with a genotype above $s$ will be selected. Geometrically, this corresponds to the alignment of the reproduction curve (orange) and the genotypic curve at time $t+1$ (green).

Further discussion of Fig.~\ref{fig:4pannelsprofiles}, including explanations of the function $A(x)$ and the terminology ``fully pulled" and ``semipulled" can be found in the next subsection and in the Summary after Theorem~\ref{thm:reprod} in SM.

\bigskip

\noindent
{\bf Low phenotypic noise.} 
If $\mu > 1$, a close inspection of Eq.~(\ref{eq:gamma_c}) reveals that
$\gamma_{c}(\mu)=0$ so that only the selection of the fittest regime persists.
In this regime, the large increment of the phenotype of an individual is primarily explained by a large increment of its genotype (see Eq. \ref{eq:tail}), and no non-trivial optimum may be found for the selection pressure~($\gamma$) when we optimize on the rate of adaptation ($v$). That is, in this case, we have $\gamma_c=0$, meaning that the optimal adaptation rate is to select a constant (independent of $N$) number of individuals with the largest phenotypes. 

\begin{figure}
    \centering
\begin{tabular}{cc}
 \includegraphics[width=0.3\textwidth]{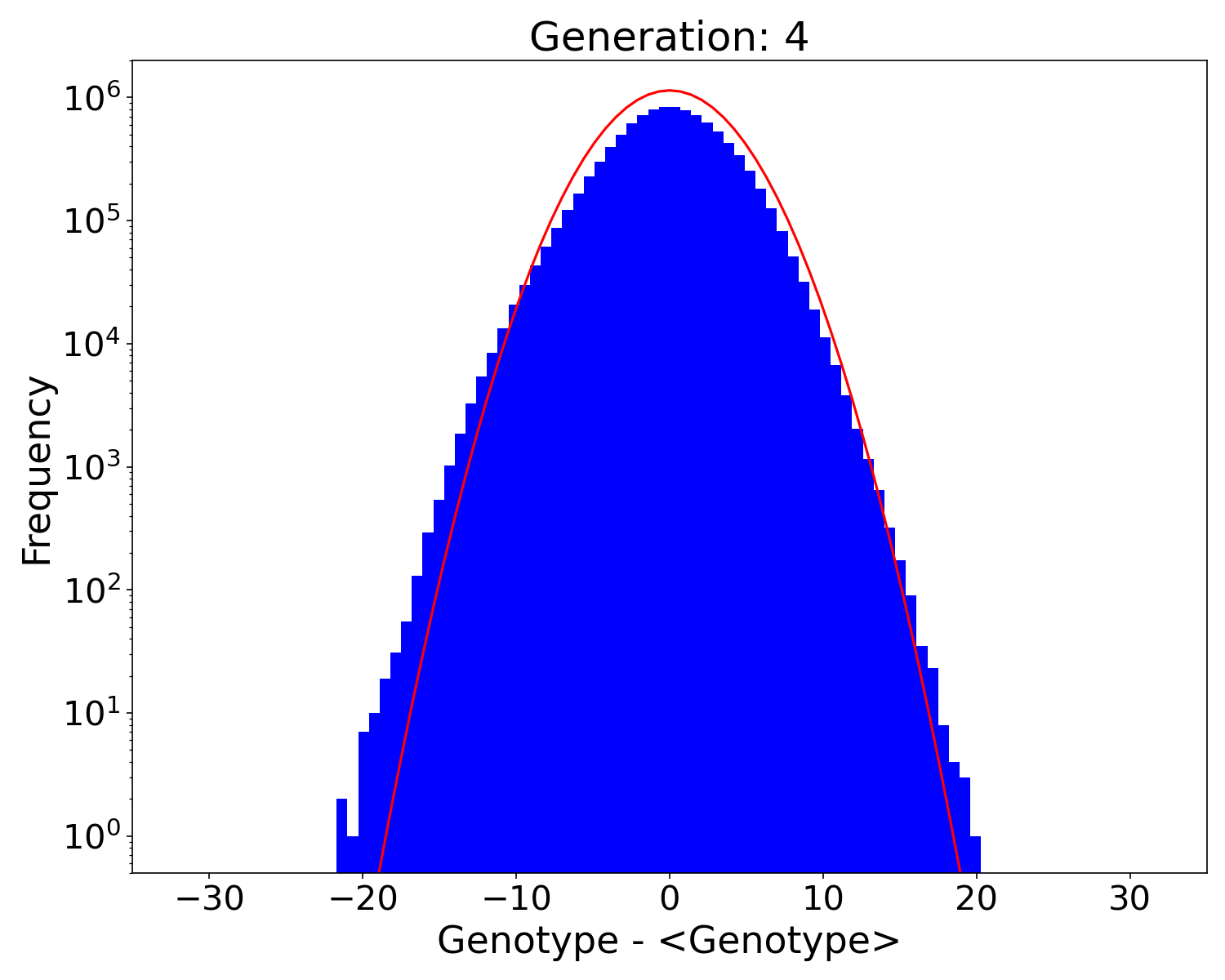} &
\includegraphics[width=0.4\textwidth]{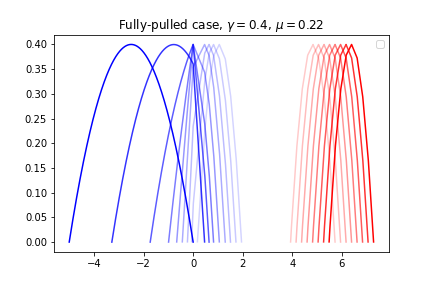} \\

 \begin{tikzpicture}[xscale = 2, yscale=6]
\draw (-1.5,.7) node {Fully pulled case};
\draw (-1.5,.63) node {$\gamma = 0.4$, $\mu = 0.4$};
\draw[->] (0,-.1) -- (0,.8);
\draw (.04,.4) -- (-.02,.4) node[left] {$\gamma$};
\draw[->] (-2.2,0) -- (1,0);

\draw [thick, color=blue!70!black] (0,0) -- (-.5,.3) -- (-1,.4) -- (-1.5,.3) -- (-2,0) node[above left] {$g$};

\draw [thick, color=green!70!black] (.5,0) -- (0,.3) -- (-.5,.4) -- (-1,.3) -- (-1.5,.0)  node[above left] {$g_1$};

\draw [thick, color = orange] (.6,0) -- (0.5,.1) -- (0,.6)  -- (-.2,.7);
\draw[color=orange] (0,.65)  node[below left] {$r$};
\draw [color=red, densely dashed] (.75,0) -- (.75,.7);
\draw[color=red] (.75,.69)  node[left] {$s$};

\draw [color=purple] (-.8,-.0) node {$\bullet$} node[below]{$x$};
\draw [color=purple] (-.3,-.0) node {$\bullet$} node[below]{$A(x)$};
\draw [color=purple] (0,-.0) node {$\bullet$} node[below right]{$A^2(x)$};
\end{tikzpicture}

&

 \begin{tikzpicture}[xscale = 2, yscale=6]
\draw (-1.5,.7) node {Semipulled case};
\draw (-1.5,.63) node {$\gamma = 0.5$, $\mu = 0.4$};
\draw[->] (0,-.1) -- (0,.8);
\draw (.05,.5) -- (-.05,.5) node[left] {$\gamma$};
\draw[->] (-2.2,0) -- (1,0);

\draw [thick, color=blue!70!black] (0,0) -- (-.1,.1) -- (-.6,.4) -- (-1.1,.5) -- (-1.6,.4) -- (-2.1,.1) -- (-2.2,0) node[above left] {$g$};

\draw [thick, color=green!70!black] (.5,0) -- (.4,.1) -- (-.1,.4) -- (-.6,.5) -- (-1.1,.4) -- (-1.6,.1) -- (-1.7,0)  node[above left] {$g_1$};

\draw [thick, color = orange] (.5,0) -- (.2,.3) -- (-.1,.6) -- (-.3,.72);
\draw [color=orange] (0,.69) node[left] {$r$};

\draw [color=red, densely dashed] (.4,0) -- (.4,.7) ;
\draw [color=red] (.4,.69) node[left] {$s$};

\draw [color=purple] (-.8,-.0) node {$\bullet$} node[below]{$x$};
\draw [color=purple] (-.3,-.0) node {$\bullet$} node[below]{$A(x)$};
\fill [color=purple] (0,.01) -- (0,-.01) -- (-.1,-.01) -- (-.1,.01) -- (0,.01)  -- cycle;
\draw [color=purple] (-.1,-.0) node [below right] {$A^2(x)$};
\end{tikzpicture} 

\end{tabular}

    \caption{
    {\it Top left panel.} Distribution of genotypes in a large population, plotted on a logarithmic scale. We observe for large populations the appearance of a deterministic genotypic profile. The profile is well  approximated by a quadratic profile (in red).
    {\it Top right panel.} Convergence to the traveling wave solution of the deterministic dynamic, started from an initial Gaussian distribution (i.e. a parabolic log-profile). Blue curves show the evolution of the first 8 steps of the genotypic profile, the red curves show steps 16 to 24. We observe that by that time, the genotypic profile has converged to a traveling wave solution.
    {\it Bottom panels. } Schematic description of the iteration of the traveling wave over one step. The blue line represents the initial traveling wave profile, the orange line the genotypic profile of the children; that is, the profile after reproduction but before selection (we only display the rightmost part of this profile). The green line corresponds to the genotypic profile of the selected children, and is a translation of the blue profile. The vertical red dotted line corresponds to the position of the selection threshold -- in particular, note that the green and orange lines are identical to the right of that threshold, indicating that genotypes to the right of the threshold will be selected. The bottom purple dots correspond to iterations of the ancestral function (see Section~\ref{sec:fullysemi}) in the stationary setting indicating the successive ancestors in the population. 
    }
    \label{fig:4pannelsprofiles}
\end{figure}

\noindent

\subsection{Fully pulled and semipulled waves}\label{sec:fullysemi}

In models of adaptation, long-term evolutionary dynamics are often described by pulled traveling waves, whose advance is controlled by rare individuals at the leading edge of the population \cite{hallatschek2011noisy}.

The critical curve $\mu\to\gamma_{c}(\mu)$, which separates the survival of the fittest and luckiest regimes,
also delineates a phase transition between what we call semipulled and fully pulled waves.
To explain the nature of the transition, we first note that the dynamics of the wave contains some partial information on the ancestral structure of the underlying  population. Recall that 
$$
g_{n}(x) \ = \  \pi\left[1 - \gamma + \sup_{y \in \R} \left( g_{n-1}(y) - |x-y| - \mu (s_n - x)_+\right)\right],
$$
where the maximization problem arises from integrating the contribution of all the population at time $n-1$ and then applying the Laplace method to extract the main contribution at $x$ from the previous generation. This overwhelming contribution by a single location is captured by the $\sup$ in the previous formula. This entails that
the ancestor of an individual at position $x$ is likely to be found at
$$
\mbox{Argmax}_{y \in \R} \left( g_{n-1}(y) - |x-y| - \mu (s_n - x)_+\right)
$$
in the previous generation.
Let us now consider an individual at distance $x$ from the extremal genotypes (right tip of the wave).  From the previous equation, we deduce that the distance of its ancestor
from the tip (of the wave in the previous generation) is given by the ancestral map
\begin{eqnarray}\label{eq:ancestor}
  A (x) =  \mathrm{Argmax}_{y \in \R} \{g(y) - |x+v-y|\},
\end{eqnarray}
where $g$ is the traveling wave. Let $A^{(n)}$ be the $n$th iteration of the ancestral map $A$. That is, the genotype of the ancestor $n$ generations backward in time is at distance $A^{(n)}(x)$ from the tip of the wave $n$ generations backward in time. It turns out that the two evolutionary regimes  dictate different behaviors for the ancestral map $A$. This is formally proved in Theorem~\ref{thm:ancestralLineages} in SM, and graphically explained in the bottom panels of Fig~\ref{fig:4pannelsprofiles} (see also the Summary after Theorem~\ref{thm:reprod} in SM). 

In the selection of the luckiest regime ($\gamma<\gamma_{c}(\mu)$), $A^{(n)}(x)$ reaches $0$ in finitely many generations so that the ancestor of any individual is directly at the tip of the wave. Thus, the wave at a given time is generated (or pulled) by few extremal individuals close to the tip and is said to be {\it fully pulled}. See bottom left panel of Fig~\ref{fig:4pannelsprofiles}. 

For the selection of the fittest regime ($\gamma>\gamma_{c}(\mu)$), $A^{(n)}(x)$ also reaches an equilibrium in finitely many generations. However, and in contrast to the luckiest regime, the maximization problem (\ref{eq:ancestor}) becomes degenerate. More precisely, iterating the ancestral map $A$ gets the ancestor closer to $0$ (i.e.~closer to the tip). After a few iterations, the maximum of the ancestral function $A$ is not attained at a single point but on an interval that we refer to as the {\it ancestral interval}. See bottom right panel of Fig.~\ref{fig:4pannelsprofiles}.  The interpretation of this  phenomenon is that the positions of ancestors of an individual are uniformly distributed on the ancestral interval after a few generations. A crucial observation is that the ancestral interval (1) contains the tip, but (2) does not contain the bulk, that is, the point where the wave is maximized. As a consequence of (2), ancestral individuals deviate substantially from the mean but in contrast to the luckiest regime, a typical ancestor is not directly at the tip, but instead it is uniformly distributed on the ancestral interval. Thus, the wave is still pulled by extremal individuals, but those extremal individuals are typically located at an intermediate location between the tip and the mean. The wave is now said to be only {\it semipulled}. 

Pulled traveling waves are commonly treated as a single universality class, in which adaptation is driven by rare individuals at the leading edge; our results show that this class itself splits into distinct semipulled and fully pulled regimes with qualitatively different ancestral structures.

\bigskip

\noindent

\subsection{ Effective population size} \label{sect:effecti-pop}

The effective population size is defined as the expected coalescence time of two distinct ancestral lineages. While our earlier analysis is based on deterministic approximations, the genealogical structure of the system retains its stochastic nature in the large population limit and random coalescence times cannot be inferred solely from the hydrodynamic limit. Consequently, the genealogical structure of the population (and thus $N_e$) highly depends on the system fluctuations.

Analyzing system fluctuations is inherently more complex, so we address this challenge using a proxy model inspired by the exponential $K$-branching random walk ($K$-BRW) developed by Brunet, Derrida, and collaborators \cite{Brunet1997,Brunet2007}. In this original model, individuals reproduce an infinite number of offspring distributed according to an exponential Poisson point process centered at the parental value. After reproduction, truncation selection is applied, retaining only the 
$K$ rightmost genotypes. This framework was originally introduced to provide the first analytical approach to studying fluctuations in F-KPP fronts, leveraging the particle system's integrability.

A recent generalization, known as the noisy exponential $K$-BRW, modifies this framework by blurring the effect of selection: offspring are reproduced as before but instead of truncation selection, individuals are randomly selected according to Gibbs sampling \cite{Cortines2018,schertzer2023relative}. For more details, see Section~\ref{sect:k-lineages} of the SM.

Although the noisy exponential $K$-BRW may initially seem distinct from our study, Section \ref{sect:k-lineages} of the SM will demonstrate that the differences are largely superficial. We show that the noisy exponential $K$-BRW retains many key features of the original model, including a similar hydrodynamic limit and the transition between semipulled and fully pulled regimes. This suggests that both models belong to the same universality class and share similar genealogical structures. By leveraging the integrability of the noisy exponential $K$-BRW and results derived in \cite{schertzer2023relative}, we compute the effective population size for this integrable model (see~\eqref{eq:Ne-below-gammac} and~\eqref{eq:Ne-above-gammac} in SM). Comparing universal quantities between the two models allows us to propose an ansatz for the effective population size in the original model:
\begin{eqnarray}
\forall \gamma<\gamma_{c}(\mu), & N_e(\mu,\gamma) \  \approx  \ \frac{1}{\mu} \label{eq:T2-strong} \\
\forall \gamma>\gamma_{c}(\mu), & N_e(\mu,\gamma) \approx \chi(\mu,\gamma) \log(N),
\label{eq:T2-weak}
\end{eqnarray}
where  $\chi(\mu,\gamma)$ is
the size of the ancestral interval and is given by Eq.~\eqref{eq: vminussigma} in Theorem~\ref{x} in SM and is in good qualitative accordance with our numerical simulations; see Fig.~\ref{fig: rateOfCv}. 
In particular, ~\eqref{eq:T2-strong} and~\eqref{eq:T2-weak} explain the non-monotonicity of the function $\mu\mapsto N_{e}(\mu,\gamma)$; indeed, it is straightforward to check that $\mu\mapsto\chi(\mu,\gamma)$ is an increasing function. As in Fig~\ref{fig:2pannelsvne}, the change of monotonicity again occurs at the critical line: in the right panel of Fig.~\ref{fig: rateOfCv}, for $\gamma=1/2$, this is close to $\mu_c = \gamma_c^{-1}(1/2)=1/3$ (see $\gamma_c(1/3)=1/2$ from~\eqref{eq:gamma_c}).
We note that the discontinuity in the theoretical predictions corresponds to a first order phase transition and that this discontinuity is smoothed out in the finite population regime.

Finally, our comparative approach allows us to extract more information on the ancestral structure of the population. Whereas $N_e$ depends only on the coalescence time of two lineages, our comparative analysis describes the random genealogy spanned by any number of lineages. In SM, we show that in the selection of the fittest regime (resp.~luckiest regime)
the genealogy should converge to a Bolthausen--Sznitman coalescent (resp.~Poisson--Dirichlet coalescent) \cite{Pitman1999}. See Section~\ref{sect:k-lineages} for more details.

\subsection{Slow convergence to the deterministic limit.}

Figure \ref{fig: rateOfCv} demonstrates that our deterministic approximations offer a reliable qualitative prediction of the stochastic model. Notably, the finite-size particle system retains a similar phenomenological structure to its deterministic counterpart, including the sharp phase transition between the \emph{selection of the fittest} and \emph{luckiest} regimes, as shown in Fig.~\ref{fig: rateOfCv}. However, the convergence to the hydrodynamic limit is observed to occur at an exceptionally slow rate.

This slow convergence is a well-documented characteristic of branching-selection particle systems~\cite{Derrida2008}. To better understand the deviations from the infinite population limit, we utilize the noisy exponential $K$-BRW introduced in SM Section~\ref{sect:k-lineages} and compute the convergence rate (in terms of $N$, where $K=N^\gamma$) in SM Section~\ref{sec:corrections}. The speed of convergence of the adaptation rate occurs at a very slow $(\log N)^{-1}$ rate in the fully pulled regime, and at an even slower rate $\log(\log(N))/(\log N)$ in the semipulled regime.

Higher order corrections affect the convergence to the limiting genotypic profile as well. Although the limiting traveling wave has a piecewise linear log-profile (as shown in Theorem~\ref{x} in SM and the bottom panels of Fig.~\ref{fig:4pannelsprofiles}), simulations indicate that the log-profile is still well approximated by a quadratic profile due to finite-population effects (see the top right panel of Fig.~\ref{fig:4pannelsprofiles}). A quadratic log-profile corresponds to a Gaussian genotype distribution, which aligns with the typical predictions of quantitative genetics models (see~\cite{RoS} and the references therein). 

\begin{figure}[H]
\begin{center}
    \begin{tabular}{cc}

    \includegraphics[width=0.4\textwidth]{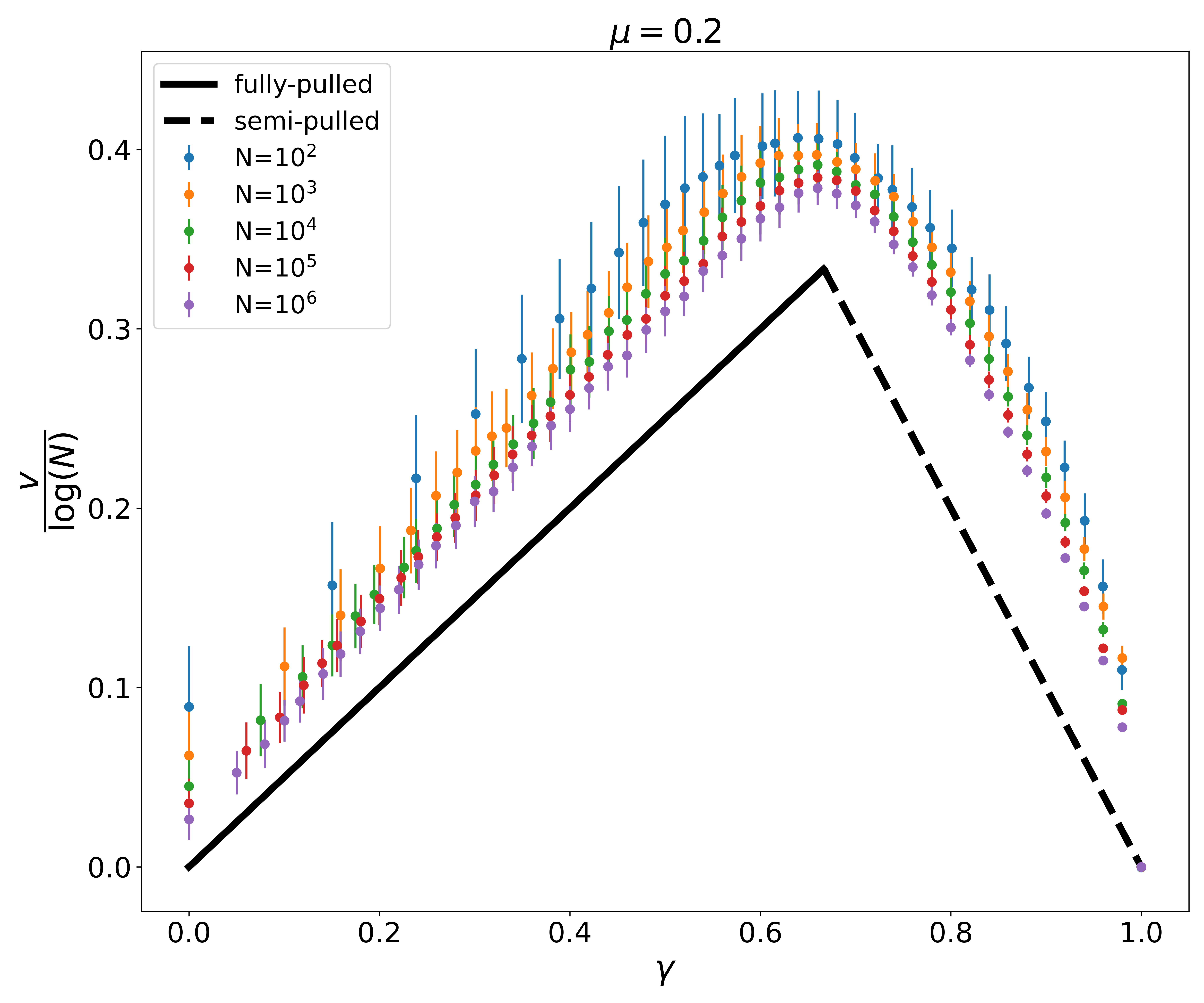}
    
    &
    
    \includegraphics[width=0.4\textwidth]{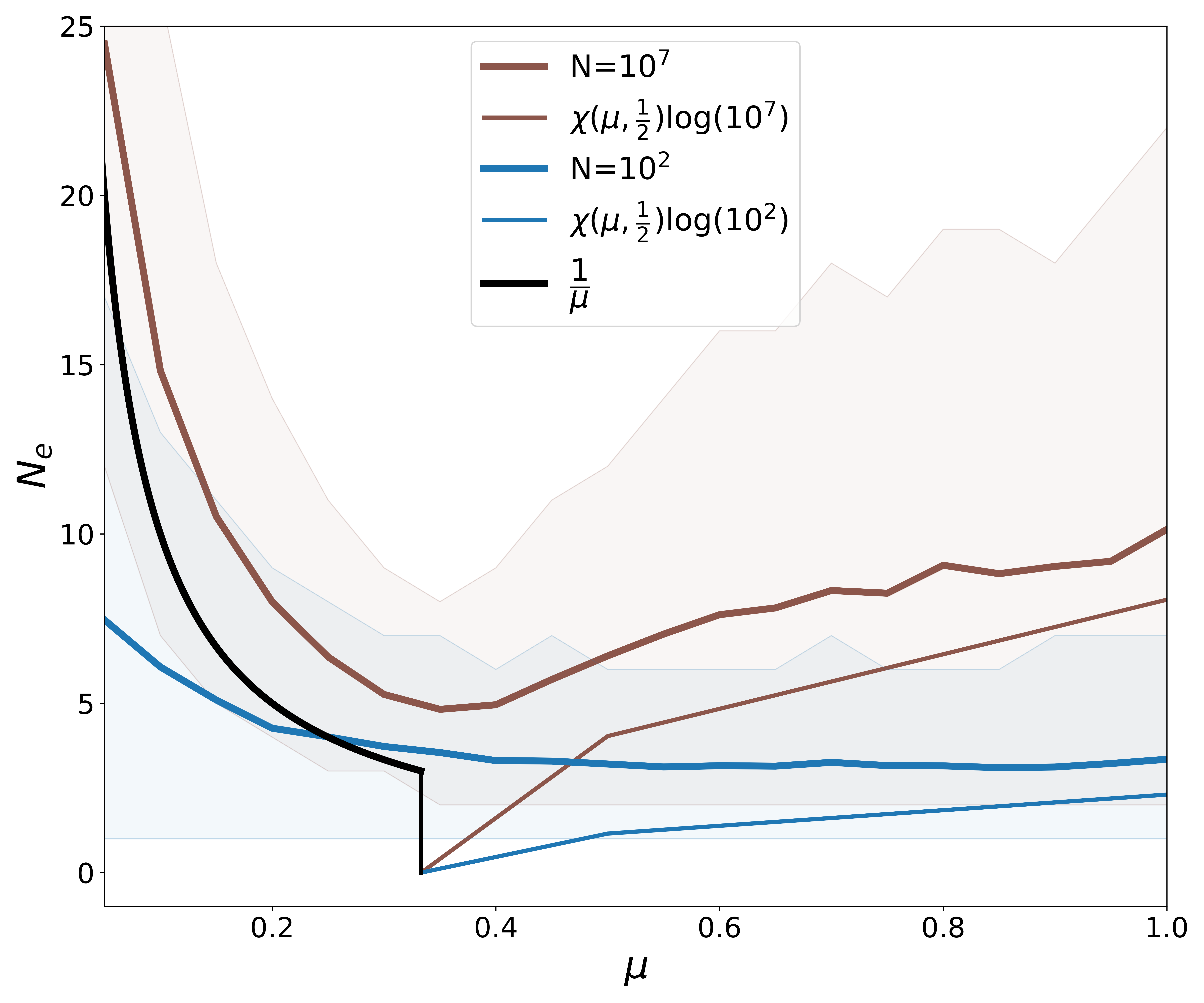} 
    \end{tabular}
\end{center}
    \caption{
    {\it Left Panel:
    } Rate of adaptation for the deterministic limiting model and its finite-size population correction, for the original stochastic model.
    The rate of convergence is notably slow.
    {\it Right Panel.}  Effective population size for $N=10^7$ (thick brown curve, average over 1000 realizations) and $N=10^2$ (thick blue curve, average over 1000 realizations) with $\gamma = 1/2$. The shaded areas show the 10th to 90th percentiles of the distribution for both cases (the middle area corresponds to the overlap). The solid black line is the theoretical approximation~\eqref{eq:T2-strong} in the luckiest regime; the thin brown and blue lines are the theoretical approximations~\eqref{eq:T2-weak} in the selection of the fittest  regime. 
    For large populations ($N=10^7$), a change of monotonicity occurs close to the predicted value, which corresponds to the end point ($\mu_c=\gamma_c^{-1}(1/2)$) of the curve $1/\mu$.}
    \label{fig: rateOfCv}
\end{figure}

\section{Discussion}

In the spirit of Fisher’s observation that ``natural selection is not evolution,'' our analysis shows that selection strength is not a monotonic driver of evolutionary efficiency. Increasing reproductive skew plays a dual role: it initially amplifies genetic differences, but beyond a critical threshold it preferentially amplifies phenotypic noise. This trade-off is not captured by short-term responses to selection \cite{lande1979quantitative,falconer1996introduction}, but only emerges on evolutionary time scales, where adaptation rate and genetic diversity arise endogenously from the interaction between mutation, selection, and phenotypic variability.

A central outcome of this interplay is the emergence of two distinct evolutionary regimes—selection of the fittest and selection of the luckiest—separated by a sharp transition. In the former, selection reliably amplifies genetic signal, and stronger selection accelerates adaptation. In the latter, overly strong selection amplifies stochastic phenotypic fluctuations, so that reproductive success is dominated by chance rather than heritable fitness. In this regime, the luckiest individuals are systematically favored over the genetically fittest, slowing adaptation and collapsing effective population size. As a result, optimal adaptation is attained at an intermediate selection strength that balances signal amplification with robustness to luck. This mechanism provides a biological realization of Goodhart’s law: when a phenotypic proxy becomes an overly optimized target, it ceases to reliably promote genetic improvement.

Although population fitness profiles remain approximately Gaussian, their variance is not imposed but self-organized, emerging from long-term evolutionary dynamics under noisy selection. The transition between the two regimes marks not only a change in adaptation speed, but also a qualitative shift in evolutionary predictability. In the luck-dominated regime, long-term evolutionary outcomes become increasingly sensitive to rare stochastic events affecting a few extreme individuals, rather than to typical genetic differences across the population.

If phenotypic noise itself is evolvable \cite{ito2009selection,vinuelas2012towards,keren2016massively,duveau2018fitness,hill2010genetic,pelabon2010evolution}, our results suggest a potential trade-off between speed and diversity. Reducing noise can accelerate adaptation, but in luck-dominated regimes it may simultaneously reduce effective population size. By contrast, when selection reliably targets genetic signal, lower noise enhances both adaptation rate and genetic diversity.

Importantly, all of the above results are robust across modeling assumptions. We observe the same qualitative behavior for different distributions of genotypic and phenotypic noise, including Gaussian and Laplace cases, and for both asexual and sexual reproduction. The existence of the two regimes and the transition between them do not rely on specific microscopic details, but instead reflect a generic consequence of selection acting on noisy phenotypic proxies. This universality suggests that the trade-off we identify should arise broadly in evolving populations subject to phenotypic variability.

The importance of luck has been widely discussed in evolutionary biology, from historical contingency at the level of species \cite{gould1989wonderful} to variation in individual lifetime reproductive success both theoretically and experimentally. A recent experimental study~\cite{zipple2025} demonstrates that in genetically identical male mice, minor early-life differences can be amplified through competitive social feedback, resulting in large inequalities in reproductive success. In addition, several theoretical studies~\cite{snyder2021time, snyder2018pluck, tuljapurkar2009dynamic, caswell2009stage, van2017lifetime, hartemink2018variance}, illustrated with empirical case studies, found that much of the variation in lifetime reproductive success can be attributed not to individual quality or fitness differences, but to luck, especially in early life. Related ideas have also been influential in the social sciences, in the context of individual financial or scientific success and inequality~\cite{merton68, cook2010winner}.

While previous work has primarily emphasized the role of chance in shaping short-term individual outcomes or macroevolutionary patterns, our results provide a complementary perspective by linking individual-level luck to population-level evolutionary dynamics. In our framework, the stochastic survival or loss of extreme genotypes induces qualitative changes in adaptation rate and effective population size over evolutionary time scales.

We believe that our results can be experimentally tested in Lenski--type experiments \cite{lenski2017experimental}. A possible idea is that the phenotype--based reproduction skewness can be implemented by setting up a challenge (e.g.~swimming through a tube in which a nutrient gradient is maintained) and allowing varying fractions of the population, based on the rank achieved in the challenge, to reproduce.

Finally, from a theoretical perspective, our findings place evolutionary dynamics within a broader classification of stochastic traveling waves. While previous work has established a zoology of fully pushed and semipushed waves \cite{Brunet2001,Berestycki2013,birzu2018fluctuations,tourniaire2021branching}, pulled waves have traditionally been viewed as a single universality class dominated by their leading edge. Our results uncover a new zoology for pulled waves, distinguishing fully pulled and semipulled regimes with sharply different ancestral structures and sensitivities to stochastic extremes. This classification clarifies how luck, selection, and noise jointly shape adaptation and genealogies.

\subsection*{Data and Software Availability}
A software implementation written in C of the model described in Section \ref{sec:theModel}, the datasets produced with that software, as well as the Python scripts used to analyze the dataset and produce the figures are freely available on Zenodo: \url{https://doi.org/10.5281/zenodo.19100553}.

\subsection*{Acknowledgments}
E.S. and Z.T. thank Colin Desmarais for helpful discussions.
F.P. gratefully acknowledges funding from Tamkeen under the NYUAD Research Institute grant CG009 to Mubadala ACCESS, as well as support from the NYUAD HPC center. E.S. gratefully acknowledges support from the FWF project PAT3816823. B.M. was partially supported by the MITI interdisciplinary program entitled 80PRIME GEx-MBB, as well as the ANR MBAP-P (ANR-24-CE40-1833) project. This work was largely completed while Z.T.~was a postdoctoral researcher at the University of Vienna.

\printbibliography

\newpage

\newpage

\setcounter{section}{0}
\renewcommand{\thesection}{S.\arabic{section}}
\renewcommand{\theHsection}{S.\arabic{section}}

\renewcommand{\theequation}{S.\arabic{equation}}
\renewcommand{\theHequation}{S.\arabic{equation}}
\numberwithin{equation}{section}

\begin{center}
{\Large \textbf{Supplemental Materials for:} Selection of the fittest or selection of the luckiest: the emergence of Goodhart’s law in evolution}
\end{center}

\section*{Outline}

In this Supplementary material, we present the mathematical analysis of the model described in Section~\ref{sec:theModel} of the main text. We first describe in more detail in Section~\ref{sec:mainResults} the properties of the sequence $(g_n,p_n)$ of genotypic and phenotypic profiles that will be proved in this text, both regarding the traveling wave solutions of the dynamic \eqref{eq: deterministicDynamic} and the genealogical properties of the underlying model.

These results are then proved in the next three sections of the Supplementary material. In Section~\ref{sec: properties} we describe some a priori results on the profile dynamic, including regularization properties and preservation of concavity. These observations are used in Section~\ref{sec:travelingWave} to identify the traveling wave solutions to the dynamic, and in Section~\ref{sect: ancestry} to describe the genealogical relationships of individuals in the model.

In Section~\ref{sec:corrections}, we provide heuristics for the rate of convergence of the adaptation rate. This convergence appears to be particularly slow, however, numerical simulations show that the phase transition observed for the deterministic dynamic of fronts is also well-marked for finite size populations. Finally, Section~\ref{sect: sexual} discusses the extensions of our results to population models with sexual reproduction, highlighting the main differences between sexual and asexual models.

\section{Main results}
\label{sec:mainResults}

Let us recall that in the main article, we introduced a finite size branching-selection population model, in which individuals give birth to a large number of children, with a genotype inherited from their parent with a random increment, and a phenotype given by a random increment of their genotype. The selection procedure acts on the phenotypes. We observed that in the large population limit, the evolution of this system is well-described by a family of phenotypic and genotypic profiles $p_n$ and $g_n$, corresponding respectively to the distribution of phenotype among the children of the $(n-1$)st generation and the distribution of genotype among the selected individuals that make the $n$th generation. We recall that the profiles $(p_n,g_n)$ evolve according to the following deterministic recursive dynamic
\begin{equation}
  \label{eq: SdeterministicDynamic}
  \begin{split}
    p_{n} (x) &= \pi\left[1-\gamma + \sup_{y \in \R}\left( g_{n-1}(y) - \min(1,\mu) (x-y)_+\right)\right]\\
    s_{n} &= \sup\{x \in \R: p_{n}(x) \geq \gamma \}\\
    g_{n}(x) &= \pi\left[1 - \gamma + \sup_{y \in \R} \left( g_{n-1}(y) - |x-y| - \mu (s_n - x)_+\right)\right],
  \end{split}
\end{equation}
with $\pi : x \in \R \mapsto x \ind{x \geq 0} - \infty \ind{x < 0}$.

Here, in Section~\ref{subsec:derivation} we formally deduce the expressions \eqref{eq: SdeterministicDynamic}. Then, in Section~\ref{subsec:mainResult} we discuss the existence and properties of traveling wave solutions to this dynamic. We observe that the behavior of these traveling waves sharply depends on the values of $\gamma$ and $\mu$, and exhibits a phase transition between two distinct behaviors, which are called \emph{semipulled} and \emph{fully pulled}, corresponding to the survival of the fittest and luckiest regimes, respectively. Section~\ref{sect: fullysemi} then describes in more detail the differences between these two regimes. In Section~\ref{subsec:GaussianNumerics} we show numerical results analogous to those of Fig.~\ref{fig:2pannelsvne}, but for the case of Gaussian, rather than Laplace distributions. Finally, Section~\ref{subsec:stateOfTheArt} compares our results to the existing state of the art. 

\subsection{Derivation of the recursive dynamics}
\label{subsec:derivation}

We derive here the expression \eqref{eq: deterministicDynamic}, showing its connection with the stochastic population dynamic we introduced. We first prove the tail estimates in \eqref{eq:tail}, before providing an heuristic argument for~\eqref{eq: deterministicDynamic}.

\subsubsection{Tail asymptotics}\label{sect:tails}
We now derive the formulas in~\eqref{eq:tail}.
Consider two independent random variables $X$ and $Y$ described by the following Laplace distributions:
\begin{equation*}
f_X(x) = \frac{1}{2}e^{-|x|};\qquad f_Y(x) = \frac{\mu}{2}e^{-\mu|x|}
\end{equation*}
with $\mu>0$. For $\mu\neq 1$ we have
\begin{align*}
    \mathbb{P}(X+Y>y\log N) = &\frac{\mu}{4}\int_{y\log N}^\infty dz \int_{-\infty}^{\infty} dx\, e^{-|x| - \mu|z-x|} \\
    = & 
    \begin{cases}
    \frac{e^{-\mu y \log N}}{2(1-\mu^2)}\left( 1 - \mu^2 e^{-(1-\mu)y\log N} \right) & y\ge 0 \\
    1 - \frac{e^{\mu y\log N}}{2(1-\mu^2)} \left(1-\mu^2e^{(1-\mu)y\log N}\right) & y < 0.
    \end{cases}
\end{align*}
Taking the logarithm, in the limit $N\to\infty$, we obtain
\begin{equation}
    \label{eqn:tail1}
    \frac{\log\left(\mathbb{P}(X+Y>y\log N)\right)}{\log N} \to -\min(\mu,1)y_+ 
\end{equation}
with the notation $a_+ = \max(a,0)$.
Strictly speaking, this derivation is valid for $\mu\neq 1$. In computing the probability on the left-hand side for $\mu=1$, one finds that the expression \eqref{eqn:tail1} also holds in this case.

To derive the second formula in~\eqref{eq:tail}, we let $x \neq 0$ and $y \neq x$, and set $\epsilon \in (0,\min(|x|,|y-x|)$. We have
\begin{align*}
 &   \mathbb{P}(|X - x\log N| < \epsilon \log N, X+Y>y\log N) &
    \\
 &   =  \frac{\mu}{4} \log N \int_{x-\epsilon}^{x + \epsilon} e^{-|a|\log N} \int_{(y- a)\log N}^\infty e^{-\mu|b|}\, db \, da &\\
    =
    &\quad 
  \begin{cases}
    \frac{1}{4} \log N e^{-\mu y \log N} \int_{x-\epsilon}^{x + \epsilon} e^{-|a|\log N + \mu a \log N} da\,
    & y-x\ge 0;\\
    \frac{1}{4} \log N\int_{x-\epsilon}^{x + \epsilon}  e^{-|a|\log N}\left(2-e^{\mu(y-a)\log N}\right)\,da & y-x<0, 
    \end{cases}
\end{align*}
where we use that $y-x-\epsilon$ and $y-x+\epsilon$ have the same sign as $y-x$. We can then compute, when $x > 0$, $y-x>0$ and $\mu \neq 1$,
\begin{equation*}
    \mathbb{P}(|X - x\log N| < \epsilon \log N, X+Y>y\log N) = \frac{e^{-(x + \mu(y-x))\log N}}{4} \frac{e^{|1-\mu|\epsilon \log N} - e^{-|1 - \mu| \epsilon \log N}}{|1-\mu|},
\end{equation*}
from which we deduce the expected asymptotic behaviour, with a similar computation holding when $\mu = 1$. Other cases, when $x < 0$ or $y-x < 0$ can be treated in the same fashion, leading to
\begin{equation}
    \label{eqn:tail2}
    \lim_{\epsilon \to 0} \lim_{N \to \infty} \frac{1}{\log N}  \log\left(\mathbb{P}( |X- x\log N| < \epsilon \log N, X+Y>y\log N)\right) =-\left(|x|+\mu(y-x)_+\right).
\end{equation}

\subsubsection{Recursive expressions for \texorpdfstring{$p_n$}{pn} and \texorpdfstring{$g_n$}{gn}}\label{sect:RecursiveExpr}

A rigorous statement and proof saying that the dynamical system in~\eqref{eq: deterministicDynamic} describes the scaling limit as $N \to \infty$ of the stochastic model we introduced is beyond the scope of this article. We give here the heuristics for this connection without going into technical details. A rigorous argument for a similar model is provided in~\cite{desmarais2025k}.

Assume that the population at time $n$ is distributed according to the profile $g_n$, that is, for a fixed location $y\in\R$, there are roughly $N^{g_n(y)}dy$ individuals whose genotypes belong to the interval $[y\log N,(y + dy)\log N]$. Each of these individuals will create $N^{1-\gamma}$ children. Thus, by independence of the offspring generation, the number of children with a phenotype greater than $z\log N$ will be given by a binomial random variable with parameters $N^{1-\gamma + g_n(y)}dy$ and $\P((X+Y) > (z-y) \log N)$. Applying~\eqref{eqn:tail1}, its expectation is close to
\[
N^{1-\gamma + g_n(y)- \min(\mu,1)(z - y)_+}dy.
\]
To find the total expected number of such children, we integrate over all locations to obtain
\[
\int N^{1-\gamma + g_n(y)- \min(\mu,1)(z - y)_+}dy \approx N^{1 - \gamma + \sup_{y \in \R} (g_n(y) - \min(1,\mu)(z-y)_+)},
\]
where the informal approximation comes from Laplace's principle. If the exponent on the right-hand side is negative, then there are no phenotypes greater than $z\log N$ among the children with high probability, explaining the presence of $\pi$ in the definition~\eqref{eq: deterministicDynamic} of $p_{n+1}$. Therefore, the total number of children of generation-$n$ parents with a phenotype larger than $z\log N$ can be approximated by
\[
N^{p_{n+1}(z)},
\]
for large values of $N$. 

Let us write $S^N_{n+1}$ for the position of the $N^\gamma$th largest child of generation-$n$ parents.
As a consequence of the previous argument and since the function $z \mapsto p_{n+1}(z)$ is non-increasing, we approximate the phenotypic threshold by
\begin{equation}\label{eq:Sntosn}
\frac{S^N_{n+1}}{\log N}\approx\sup\{z \in \R : p_{n+1}(z) \geq \gamma\} = s_{n+1}
\end{equation}
(see also~\eqref{eq: deterministicDynamic}).

Finally, the formula for the function $g_{n+1}$ is obtained in the same fashion as $p_{n+1}$. We observe that the individuals alive at time $n+1$ with a genotype close to $x\log N$ are descendants of parents in generation $n$ with a genotype close to $y\log N$ for some $y$, and all have a phenotype larger than $S^N_{n+1}$. Then, using~\eqref{eq:Sntosn} and \eqref{eqn:tail2}, for each genotype around the location $y\log N$, the expected number of children around $x\log N$ will be close to
\[
  N^{ 1 - \gamma + g_{n}(y) - |x-y| - \mu(s_n -x)_+ }dydx.
\]
Similarly to the phenotypic profile, integrating over all locations and using Laplace's principle we find that the number of genotypes around location $x\log N$ at time $n+1$ is approximately
\[
N^{g_{n+1}(x)}dx,
\]
where $g_{n+1}$ is given by~\eqref{eq: deterministicDynamic}.

\subsection{The traveling wave}
\label{subsec:mainResult}

We say that a function $g$ is a \textit{traveling wave} for the dynamic \eqref{eq: SdeterministicDynamic} with speed $v$ if, starting from $g_0 = g$, for all $n\in\N$ we have
\begin{equation}
  \label{eqn:StravelingWave}
  g_n : x \mapsto g(x-nv).
\end{equation}
In other words, the dynamic has the effect of shifting the genotypic profile by $v$, without deformation. Given the definition of profiles, we refer to the \emph{support} of the traveling wave $g$ as $g^{-1}(\R_+)$, the set of points $x$ such that $g(x) \geq 0$. That is, $x$ is in the support of $g$, if there will be individuals in the neighbourhood of $x \log N$ with high probability, when $N$ is large enough.

Our main result is stated in the theorem below, which shows existence and uniqueness of a traveling wave solution to~\eqref{eq: SdeterministicDynamic}. The proof of this theorem can be found at the end of Section~\ref{sec:travelingWave}.
\begin{thm}
\label{thm:main}
For all $\gamma \in (0,1)$ and $\mu>0$ with $\mu\neq 1$, we set $k = \floor{1/\mu}$. For the dynamic described in~\eqref{eq: SdeterministicDynamic}, there exists a unique up to translation, concave traveling wave $g$ with compact support. This traveling wave has speed $v(\gamma,\mu)$ given by
\[
  v(\gamma,\mu) = \begin{cases}
    \frac{2\gamma}{k(2 - (k+1)\mu)} & \text{ if } \gamma < \gamma_c(\mu)\\
    1-\gamma & \text{ if } \gamma > \gamma_c(\mu)
    \end{cases},
    \quad \text{ where } \gamma_c(\mu) = \frac{k}{k+1} \left( 1 - \frac{\mu}{2 - k\mu}\right).
\]
In particular, the function $\gamma \mapsto v(\gamma,\mu)$ is increasing on the interval $(0,\gamma_c]$ and decreasing on $[\gamma_c,1)$, whereas the function $\mu \mapsto v(\gamma,\mu)$ is increasing.
\end{thm}

It is worth mentioning that in the above theorem, if $\mu > 1$ then $k =0$. Therefore, $\gamma_c(\mu) = 0$ for all $\mu > 1$. This indicates that if the tail of the phenotypic contribution is light enough, then the phase transition between the selection of the fittest and luckiest regimes does not occur. More precisely, the population is always in the selection of the fittest regime, and decreasing the value of $\gamma$ increases the rate of adaptation of the population. This behavior is a consequence of the fact that in this regime, the phenotypic value of an individual is very close to its genotypic value, so the phenotypic and genotypic selection procedures become essentially undistinguishable.

\begin{rmk}
If $\mu = 1$, then the law of $X/(c\log N)$ conditionally on $X + Y > c \log N$ converges in distribution to a uniform random variable on $[0,1]$, therefore the relationship between phenotype and genotype of an individual ceases to be well-concentrated around a deterministic value. We do not treat this limiting case in the present paper although the formula we obtained can be prolonged by continuity at $\mu = 1$.
\end{rmk}

The proof of Theorem~\ref{thm:main} is based on the explicit construction of the traveling waves associated to the parameters $(\gamma,\mu)$. The explicit solution is stated in Theorem~\ref{x}. The traveling wave will typically be formed as a concave, piecewise linear function whose maximum is $\gamma$. The slope near the right edge of the traveling wave is either $-1$ or $\mu - 1$ depending on whether $\gamma > \gamma_c$ or $\gamma <\gamma_c$. This slope relates to the exponential growth of the size of the population with a genotype at distance smaller than $x \log N$ from the largest genotype of a given generation, for $x$ small enough. 

As stated in the theorem, for a fixed value of $\mu$, the speed of the traveling wave (corresponding to the rate of adaptation of the population) takes its maximum at a critical value of $\gamma = \gamma_c(\mu)$ (corresponding to an optimal choice of the selection pressure). The function $\mu \mapsto \gamma_c(\mu)$ is drawn on the right panel of Fig.~\ref{fig: rateOfCvSM} in Section~\ref{sec:corrections}, together with empirical estimates of the optimal choice of $\gamma$ for various finite size population models.

In particular, if $\mu> 1$, that is, if a large increment of phenotype in an individual is primarily explained by a large increment of its genotype, the optimal selection pressure is obtained as $\gamma = 0$, i.e. a maximal selection pressure. In this situation, the optimal dynamic for the population is to select at every step a constant (independent of $N$) number of individuals with the largest phenotype at each generation to maximize the rate of adaptation of the population. The maximal genotype in the population will increase by $\log N$ in each generation, which is similar to the Brunet-Derrida behavior obtained for the exponential model of branching random walk with selection~\cite{Brunet2006,Mallein2018} introduced in Section~\ref{subsec:stateOfTheArt}.

On the other hand, if $\mu < 1$, that is, if individuals can have a very high phenotype without having a large genotype, then there exists a non-trivial optimum for the selection pressure at $\gamma = \gamma_c$. This selection pressure gives an optimal rate of adaptation of the population, by ensuring that the selection step keeps individuals with high genotypic value and makes them create a large enough number of children so that their characteristics are transmitted to their descendants.

\subsubsection*{Convergence to the traveling wave}
\label{subsec:waves}

Theorem~\ref{thm:main} shows the existence and uniqueness of a traveling wave solution to the dynamic~\eqref{eq: SdeterministicDynamic}. Proving that, starting from an arbitrary  genotypic profile (satisfying some conditions), the dynamic converges to the traveling wave solution is out of the scope of this paper. However,  simulations in the top right panel of Fig.~\ref{fig:4pannelsprofiles} in the main text indicate that convergence should indeed occur, and therefore a detailed analysis of the traveling waves is needed to describe the long term behavior of the dynamical system~\eqref{eq: SdeterministicDynamic}. 

\subsection{Fully pulled and semipulled waves}\label{sect: fullysemi}

Assuming that $\mu < 1$, the properties of the traveling wave solution are markedly different depending on whether $\gamma > \gamma_c$ or $\gamma < \gamma_c$. These differences can be explained through the genealogical relationships and selection properties that appear at the front of the population. When we refer to the front of the population, we mean the genotypes within a positive but not too large distance from the rightmost genotype (on the logarithmic scale). By the tip of the profile we understand particles located at or very near the rightmost position (on the logarithmic scale). Our terminology for the traveling waves refers to the fact that the front of the fully pulled wave is generated by the tip (i.e.~genotypes at the front have parents at the tip), and in the semipulled case the front is generated by parents at the front but not necessarily at the tip.

\subsubsection*{Ancestry}

Our method of finding the most likely location of the parent of a genotype at a given location (for
example at the front) is to study $A$, the ancestral function of the process.
Suppose that $g$ is a traveling wave solution to~\eqref{eq: SdeterministicDynamic} with speed $v$ and with $g^{-1}(\R_+) = [L,0]$ for some $L<0$. We define
\[
  A : x \in [L,0] \mapsto \mathrm{argmax}_{y \in \R} \{g(y) - |x+v-y|\},
\]
where for any concave function $u : \R\to \R \cup \{-\infty\}$, $\mathrm{argmax}_{y\in\R} \{u(y)\}$ returns the \emph{smallest} real number $y$ such that $u(y) = \max_{z \in \R} u(z)$. Given the heuristics behind the definition of the dynamic \eqref{eq: SdeterministicDynamic}, we see that $A(x)$ corresponds to the distance from the tip of the genotypic value of a typical parent of an individual at distance $x$ from the tip. Note that $A^j(x)$ then describes the distance from the tip for a typical ancestor $j$ generations in the past, in a population distributed according to the traveling wave.

For a better illustration of the difference between fully pulled and semipulled regimes, we also introduce the function $A^+$, as 
\[
A^+ : x \in [L,0] \mapsto \mathrm{Argmax}_{y \in \R} \{g(y) - |x+v-y|\},
\]
where for any concave function $u : \R\to \R \cup \{-\infty\}$, $\mathrm{Argmax}_{y\in\R}\{u(y)\}$ returns the \emph{largest} real number $y$ such that $u(y) = \max_{z \in \R} u(z)$. Now the theorem below says that, when $\gamma<\gamma_c$ (fully pulled case), then a typical ancestor of an individual finitely many generations in the past, will be at the very tip of the population. On the other hand, if $\gamma>\gamma_c$ (semipulled case), a typical ancestor will be located in an interval at the front. In this case $A^j(x)$ and $(A^+)^j(x)$ will not agree as $j$ gets large, and the length of the interval will be given by $(A^+)^j(x)-A^j(x)$ for some large $j\in\N$. We will prove this theorem in Section~\ref{sect: ancestry}.

\begin{thm}
\label{thm:ancestralLineages}
The functions $x \mapsto A(x)$ and $x \mapsto A^+(x)$ are non-decreasing. Moreover:
\begin{itemize}
    \item If $\gamma < \gamma_c$, then $A^j(x) = {A^+}^j(x) = 0$ for all $j$ large enough;
    \item If $\gamma > \gamma_c$, then there exists $c_1 > 0$ such that $A^j(x) = -c_1$ and ${A^+}^j(x) = 0$ for all $j$ large enough. 
\end{itemize}
\end{thm}

\begin{rmk}
Let $\gamma > \gamma_c$, we write $g$ for the associated traveling wave. We observe that the parameter $c_1$ defined above, corresponding to the lowest fitness ancestor that can give birth to one individual at the front of the population, verifies $L<-c_1$ and, when $\mu<1$ then $g(-c_1) < \gamma$ (these follow from the later statements Theorem~\ref{x} and Lemma~\ref{lem: Ax}). This means that, with sufficient phenotypic noise, the ancestors of the fittest individuals come neither from the bulk of the process (which are the majority of individuals, whose fitness is very close to $\sup\{x \in \R : g(x) = \gamma\}$) nor from the tip.
\end{rmk}

In other words, in the selection of the luckiest regime (fully pulled wave, $\gamma <\gamma_c$), the whole population is generated by the tip of the process a finite number of generations backward in time with high probability. In this regime, the individuals with highest genotype at a given generation will generate the whole rest of the population within a finite number of generations. On the other hand, in the selection of the fittest regime (semipulled wave, $\gamma > \gamma_c$), the population is generated by a group of ancestors consisting of the $N^{c_1}$ individuals with the highest genotypes. Coalescences within this group of ancestors then occur at a much slower, logarithmic rate, leading to the age of the most recent common ancestor of the population being of order $\log N$, compared to a constant in the other regime.

\subsubsection*{Selection properties}

The next important property, which distinguishes the fully pulled and semipulled traveling waves is the following. Let us consider the set of genotypes at the moment when reproduction has already happened, but selection has not. Then in the fully pulled regime ($\gamma<\gamma_c$), the best (largest) genotypes do not survive the selection step, whereas in the semipulled case ($\gamma>\gamma_c$) they do. 

In order to state this result precisely, we need to introduce the reproduction profile. Let $g$ be a traveling wave solution with speed $v$, and assume $g^{-1}(\R_+) = [L,0]$ for some $L<0$. Then, since $g$ is a traveling wave, if $g_0=g$ then $g_1(x)=g(x-v)$, and $g_1^{-1}(\R_+) = [L+v,v]$. Now similarly to the heuristics given for the dynamics~\eqref{eq: SdeterministicDynamic}, we define the reproduction profile
\[
r(x) := \pi\left[1 - \gamma + \sup_{y \in \R} \left( g(y) - |x-y|\right)\right],
\]
which describes the log-density of genotypes after reproduction and before selection. Note that we have $g_1(x)\leq r(x)$ for all $x\in\R$. The result below says that, if $\gamma<\gamma_c$, then the reproduction profile is strictly above the genotypic profile $g_1$ everywhere in the support of $g_1$. In particular, the right edge of the support of the function $r$ is to the right of $v$, which is the right edge of the support of the function $g_1$. That is, the largest genotypes after reproduction (the ones to the right of $v$) do not survive selection. On the other hand, when $\gamma>\gamma_c$ then the functions $g_1$ and $r$ agree on an interval at the front of these profiles (the best genotypes are among the selected ones). We prove this theorem at the end of Section~\ref{sec:travelingWave}.

\begin{thm}\label{thm:reprod}
Recall that $r$ is the log-density profile of genotypes of children of a population with genotypic profile $g$, and the genotypic profile of the surviving children is $g_1$. Then:
\begin{itemize}
\item 
  If $\gamma<\gamma_c$ then $r(x) > g_1(x)$ for all $x$ in the support of $g_1$.
\item
If $\gamma>\gamma_c$ then $r(x)=g_1(x)$ for all $x$ close enough to $v$, the right edge of the support of $g_1$.
\end{itemize}
\end{thm}

\subsubsection*{Summary}
We now provide illustrations and a summary of the properties, which are stated in the theorems above, and which distinguish the fully pulled and semipulled waves. In order to do so, let $g$ be a traveling wave solution to~\eqref{eq: SdeterministicDynamic} with speed $v$, and let $s:=s_1$ be the phenotypic threshold at generation~$1$. Note that $s$ is also the distance of the phenotypic threshold from the rightmost genotype in any generation, because of the stationarity of the profile. Let us first return to the bottom panels of Fig.~\ref{fig:4pannelsprofiles} in the main text, which we describe in the context of the  results above. 

{\it Left Panel.} Fully pulled wave. The fact that the phenotypic threshold (red segment) $s$ is to the right of the reproduction profile (orange) means that any genotype that survives selection needs to perform a large phenotypic jump (of order $\log N$) to get to the right of $s$. Since the probability of such jumps is small, and the number of genotypes near the tip of the reproduction profile is not large enough, the very best genotypes do not survive selection. That is, the genotypic profile in the next generation (green) is to the left of the reproduction profile (orange). 

Notice furthermore, that the right edge of the green curve is at $v$ (the speed of the traveling wave) and the rightmost ``break point" of the blue curve, where there is a change of slope, is at $-v$.
The purple dots illustrate Theorem~\ref{thm:ancestralLineages}: if we sample a genotype from a location $x$, then the most likely location of its parent (relative to the tip) is $A(x)>x$. In the example in the figure, $A(x)\in(-v,0)$, and the ancestral line reaches the tip in $2$ generations: $A^2(x)=0$. It is true in general, that if we sample a genotype uniformly at random from the front of the population (more precisely from the interval $[-v,0]$), then the most likely location of its parent will be at $0$ (i.e.~at the tip of the previous generation).

{\it Right Panel.} Semipulled wave. 
The right edge of the green curve is at $v$ and the rightmost ``break point" of the blue curve is at $s-v$.
After reproduction (orange), all extreme genotypes to the right of the phenotypic threshold $s$ (red segment) are selected to the next generation: the orange and green profiles coincide on the interval $[s,v]$. 

The purple dots illustrate the second case of Theorem~\ref{thm:ancestralLineages}: This is again an example, where the parent of a genotype at location $x$ has relative location $A(x)\in(-v,0)$ from the tip. However, in contrast to the fully pulled case, the relative location of the grandparent from the tip can be anywhere on the interval $(s-v,0)$; that is, $(A^+)^2(x)-A^2(x) = v-s > 0$.
In general, if we sample a genotype uniformly at random from the interval $[-v,0]$ (relative to the tip), then the most likely location of its parent will be on the interval $[s-v,0]$. The log-density of this interval is given by the first linear segment of the blue curve, whose slope is equal to $-1$.

In the table below, we give a summary of the main properties of the investigated traveling waves. The facts that in the fully pulled case ($\gamma<\gamma_c$) we have $v < s$, and in the semipulled case ($\gamma>\gamma_c$) we have $v>s$ are stated and proved in Section~\ref{sec:travelingWave}. 
\begin{table}[ht]
\caption{Comparison of the phenomenological picture between the fully pulled and semipulled regimes.}
\begin{center}	
\begin{tabular}{|m{4.5cm} | m{5.5cm} | m{5.5cm}|}
\hline
& \textbf{Fully pulled wave} & \textbf{Semipulled wave} \\
\hline
& $\gamma < \gamma_c$, $v < s$ & $\gamma > \gamma_c$, $v > s$  \\
\hline
\textbf{Stationary profile} & & \\
\hline
Speed & increasing function of $\gamma$ & decreasing function of $\gamma$ \\
\hline
Slope of $g$ at the front & $\mu-1$ & $-1$ \\
\hline
\textbf{Ancestry} & & \\
\hline
Ancestral line & located near the tip &  located away from both the tip and the bulk \\
\hline
Number of potential parents of the tip & $N^{o(1)}$ & $N^{v-s + o(1)}$ \\
\hline
\textbf{Selection} & & \\
\hline 
Phenotypic threshold & To the right of the rightmost selected genotype ($s> v$). &To the left of the rightmost selected genotype ($s<v$). \\
\hline
Probability of selection of the largest genotype & converges to 0 & converges to 1\\
\hline
\end{tabular}	
\end{center}
\label{tab:phenomenology}
\end{table}

\subsection{Numerical results for Gaussian distributions}
\label{subsec:GaussianNumerics}

We now present our numerical results for the case $\alpha=2$; that is, when both the genetic noise distribution $f_X$ and the phenotypic noise distribution $f_Y$ are Gaussian. Our simulations indicate that the two distinct regimes discussed above also appear in the Gaussian case, supporting our claims regarding universality.

Figure~\ref{fig:2pannelsvne_Gaussian} shows the normalized speed of evolution $V/\sqrt{\log(N)}$ and the effective population $N_e$. 
As in Fig.~\ref{fig:2pannelsvne}, both quantities are estimated through numerical simulations using a population size $N=10^5$.

The qualitative shapes of $V$ and $N_e$ as a function of $\mu$ and $\gamma$ remain the same as in the case $\alpha=1$ discussed in the main text. In particular, the rate of adaptation as a function of the selection pressure $\gamma$ at a fixed phenotypic noise level $\mu$ is non--monotonic, and there exists a clear functional dependence $\mu \mapsto \gamma_c(\mu)$ of the selection pressure $\gamma_c$ that maximizes the adaptation rate at any given value of $\mu$ (Fig.~\ref{fig:2pannelsvne_Gaussian}, left panel). Likewise, the effective population $N_e$, which measures the population's genetic diversity, has a non--monotonic dependence on the phenotypic noise level $\mu$ for a fixed selection pressure $\gamma$: at high noise levels (small $\mu$) the diversity decreases with decreasing noise, but on further reducing the noise the diversity increases again. Contrary to the $\alpha=1$ case, $\gamma\mapsto\gamma_c^{-1}(\gamma)$ does not perfectly track the minimal diversity.

\begin{figure}[h]
\begin{center}
    \begin{tabular}{cc}
    \includegraphics[width=0.45\textwidth]{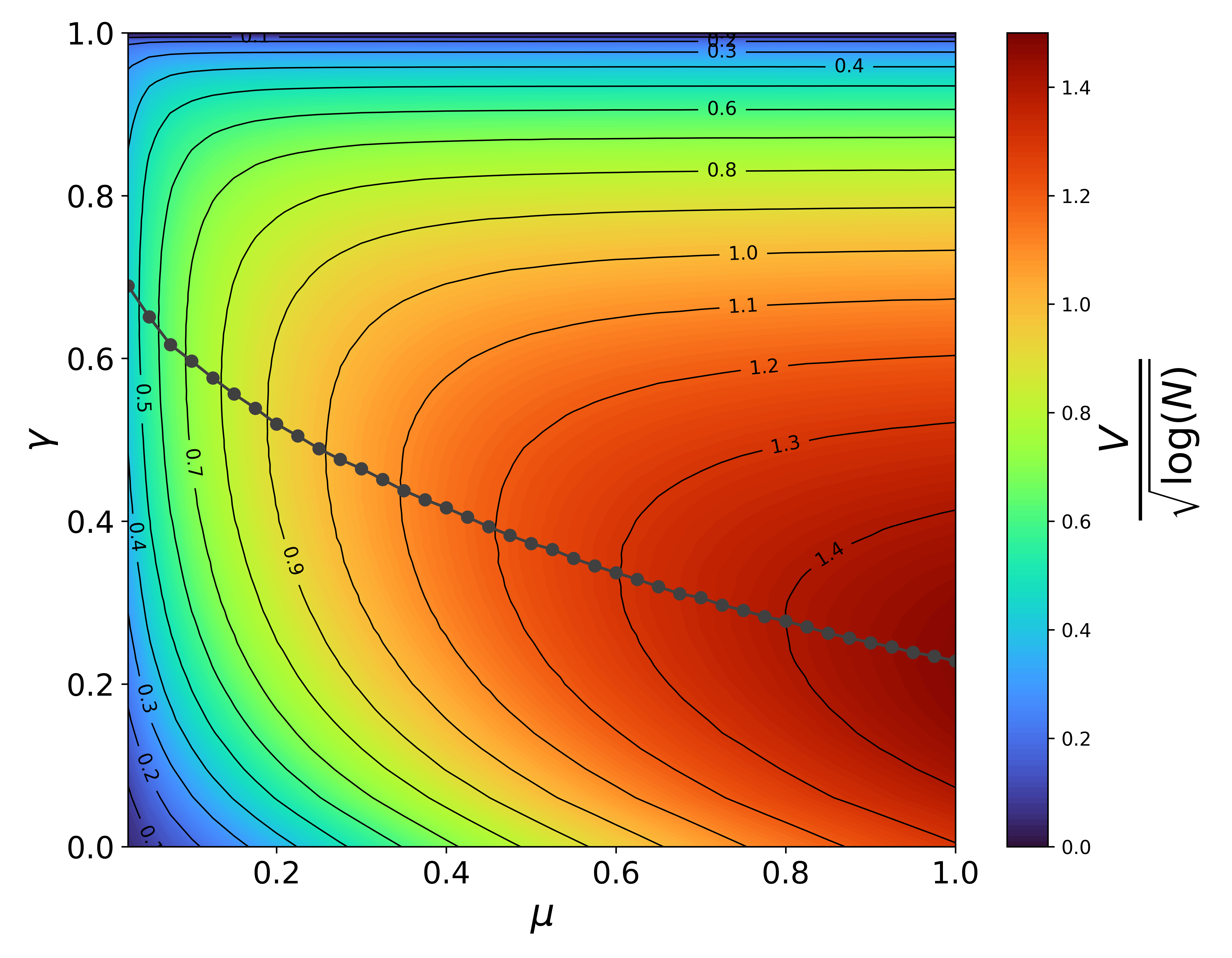}  & \includegraphics[width=0.45\textwidth]{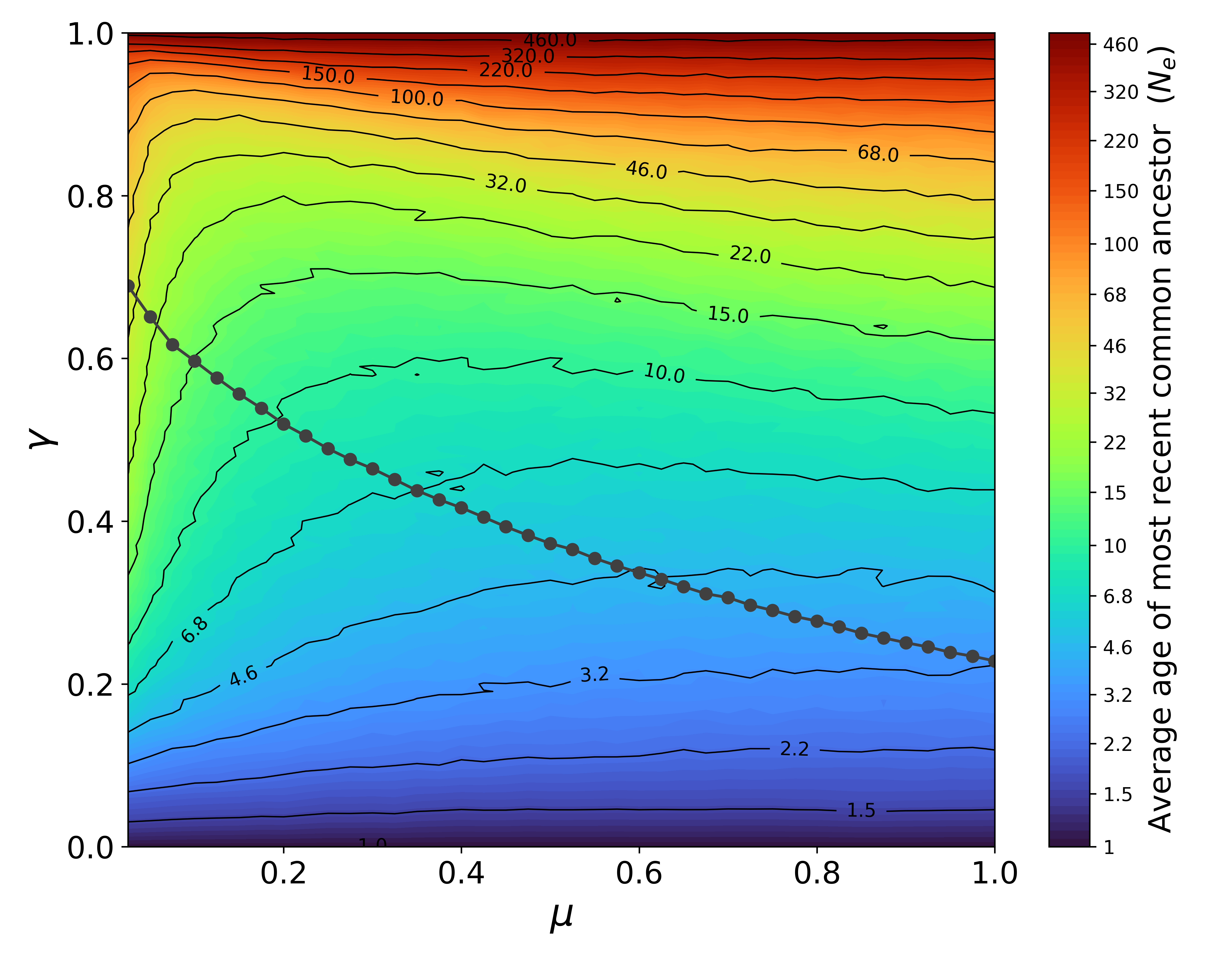}
    \end{tabular}
\end{center}
\caption{As in Fig.\ref{fig:2pannelsvne} but for the case $\alpha =2$ (Gaussian distribution for phenotypic and genotypic distributions). The general patterns appear to be very similar to the case $\alpha=1$ (Laplace distribution).}
\label{fig:2pannelsvne_Gaussian}
\end{figure}

\subsection{Related mathematical literature}
\label{subsec:stateOfTheArt}

A large number of population models for natural selection have been introduced over the years. These population models usually present a \emph{fitness}, identified as a real value, that measures the ability of an individual to survive and produce a large offspring. Among these models one can find the class of \emph{branching-selection particle systems}. These processes are defined as Markovian particle systems, in which each individual gives birth, independently to one another, to children with a fitness obtained as a random modification of their own. An external operation is then undertaken to keep the total number of individuals at each generation of roughly the same size.

One of the first of these models to be introduced was the branching random walk with absorption. In this model, individuals reproduce creating children with fitness given by an i.i.d.~copy of a point process shifted by their own. Children are then killed if their fitness falls below a given threshold. Kesten \cite{Kesten1978} showed this process to be either subcritical, critical or supercritical, i.e.~the population either grows extinct exponentially fast, polynomially fast, or survives and grows exponentially fast with positive probability.

To control the size of the population more accurately, Brunet and Derrida \cite{Brunet1997} introduced the so-called \emph{$N$-branching random walk}. In this model, the reproduction law is identical to the previous one, each individual reproduce independently by creating an identically distributed cloud of children around their position, but at each step only the $N$ rightmost children are selected to reproduce in the next generation. As a result, the total size of the population remains constant. By comparison to the KPP equation with a cutoff and the study of exactly solvable models by Brunet, Derrida, Mueller and Munier \cite{Brunet2006,Brunet2007}, they conjectured the so-called \emph{Brunet-Derrida behaviour} of the $N$-branching random walk, i.e.~that the speed $v_N$ of adaptation of the population with $N$ individuals converges as $N \to \infty$ to the speed of adaptation of the population without selection at the slow rate
\[
  v_N - v_\infty \sim \frac{-\theta}{(\log N)^2}, \text{ where $\theta$ is an explicit positive constant.}
\]
This result was later proved by Bérard and Gouéré \cite{Berard2010} for binary branching, then extended by Mallein~\cite{Mallein2018} to general branching.

The population model we introduced is analogue to an $N$-branching random walk --and more specifically to the exponential exactly solvable model introduced in \cite{Brunet2007}-- when $\gamma \to 0$. In this case, the population is capped to a value $c N$ for some $c < 1/2$, and each individual creates $\floor{1/c}$ children independently of one another. If the selection only keeps the individuals with the largest genotypic value (i.e. in the limit $\mu \to \infty$), we recover an $N$-branching random walk in which the displacements are given by i.i.d. Laplace steps. The case $\mu > 1$ corresponds to a situation when individuals are not necessarily selected if they are among the $N$ rightmost, but their probability of being selected decays exponentially fast if they are far from the maximal fitness. A solvable model in this universality class was studied by Cortines and Mallein \cite{Cortines2017}, that showed no noticeable difference in the evolution of the population due to this randomization.

For $\gamma > 0$, using extreme value theory, the evolution of the population we consider becomes closer to the solvable models studied in \cite{Brunet2006,Cortines2017,Cortines2018}, corresponding respectively to the cases $\mu = \infty$, $\mu > 1$ and $\mu < 1$. For $\gamma > 0$ and $\mu < 1$, Schertzer and Wences \cite{schertzer2023relative} studied a model close to the one we consider.
We define and analyze this exactly solvable model in detail in Section~\ref{sect:k-lineages}. We find that it exhibits the same phenomenological picture as our original model, allowing for a comparative approach to conjecture genealogical and convergence rate properties.

\section{Some a priori properties of the deterministic dynamics}
\label{sec: properties}

In this section, we present some properties of the dynamic of $(g_n,p_n)$ described in \eqref{eq: SdeterministicDynamic}. We first remark that this dynamic preserves concavity, and that $\sup g_n = \gamma$ for any $n \geq 1$. We then use this latter observation to give an alternative description of the dynamic of the genotypic profile, defining $s_n$ as the unique real number such that $\sup g_n = \gamma$.

We begin with the following straightforward observation.
\begin{lem}
\label{lem:concave}
Let $f, g$ be two concave functions $\R \to \R \cup \{-\infty\}$. The function
\[
  x \mapsto \sup_{z \in \R} f(z) + g(x-z)
\]
is concave.
\end{lem}

This result can be viewed as a property of a tropicalized version of the convolution. Similarly to the fact that log-concavity is preserved by the convolution operation, concavity is preserved by the tropicalized convolution operation.

\begin{proof}
Let $t \in (0,1)$. Using the concavity of $f$ and $g$, for $x, x' \in \R$ we have
\begin{align*}
  \sup_{z \in \R} f(z) + g(tx + (1-t)x'  - z) &= \sup_{(z,z') \in \R^2} f(tz+(1-t)z') + g(t(x-z) + (1-t)(x'-z'))\\
  &\leq t \sup_{z \in \R} f(z) + g(x-z) + (1-t) \sup_{z' \in \R} f(z') + g(x'-z'),
\end{align*}
which shows the result.
\end{proof}

By Lemma \ref{lem:concave}, the concavity of $g_n$ is preserved by the dynamic.
\begin{prop}
\label{prop:regularityG}
Let $g_0$ be a concave function, and $(g_n,p_n)$ defined recursively by \eqref{eq: SdeterministicDynamic} for all $n \geq 1$. For all $n \geq 1$, $p_n$ and $g_n$ are concave.
\end{prop}

\begin{rmk}
This result justifies the choice of only considering concave traveling waves. If the population starts from a log-concave initial distribution of genotypes, then the distribution of genotypes remains log-concave at all steps.
\end{rmk}

\begin{proof}
Take $n \geq 1$. Assuming that $g_{n-1}$ is concave, then by Lemma~\ref{lem:concave}, the function
\[
  \bar{p}_n : x \mapsto 1-\gamma + \sup_{y \in \R} g_{n-1}(y) - \min(1,\mu)|x-y|
\]
is concave, using the concavity of $z \mapsto - \min(1,\mu)|z|$. Then, using that $\pi$ is concave and increasing, we conclude that $p_n = \pi \circ \bar p_n$ is concave. Similarly,
\[
  \bar{g}_n : x \mapsto 1 - \gamma + \sup_{y \in \R} g_{n-1}(y) - |x - y| - \mu(s_n - x)_+,
\]
is concave, using Lemma~\ref{lem:concave} again and the fact that $x \mapsto - \mu (s_n-x)_+$ is concave. As a result, $g_n = \pi \circ \bar{g}_n$ is concave.
\end{proof}

For any $g : \R \to \R \cup \{-\infty\}$ and $\xi \in \R$, we define the functional
\begin{equation}\label{eq: Phis}
\Phi_\xi[g] : x \mapsto  1 - \gamma + \sup_{y \in \R} g(y) - |x-y| - \mu (\xi-x)_+.
\end{equation}
We remark that for all $n \in \N$, we have $\pi \circ \Phi_{s_{n}} [g_{n-1}] = g_n$. Therefore, defining the value of $s_n$ from $g_{n-1}$ would allow us to rewrite the dynamic \eqref{eq: SdeterministicDynamic} without reference to the phenotypic profile. This is the objective of the following lemma.

\begin{lem}\label{lem: s}
Let $g$ be a function $\R \to \R \cup \{-\infty\}$. The function
\[
F : \xi \mapsto \sup_{x \in \R} \Phi_\xi[g](x)
\]
is non-increasing.
\end{lem}

\begin{proof}
By immediate computations, we have
\[  \sup_{x \in \R} -|x-y| - \mu (\xi-x)_+ = \begin{cases}
    - \mu (\xi-y)_+ & \text{ if $\mu < 1$}\\
    - (\xi-y)_+ & \text{ if $\mu > 1$}.
  \end{cases}
\]
Therefore, for $\xi \in \R$, we have
\[
  F(\xi) = 1 - \gamma + \sup_{(x,y) \in \R^2} g(y) - |x-y|-\mu(\xi-x)_+ = 1 - \gamma + \sup_{y \in \R} g(y) - \min(\mu,1)(\xi-y)_+.
\]
Using that $\xi \mapsto - \min(\mu,1)(\xi-y)_+$ is non-increasing for all $y \in \R$, we conclude that $F$ is non-increasing as well.
\end{proof}

We are now able to state the alternative construction of the deterministic dynamic of the genotypic profile.
\begin{prop} \label{prop: sngn}
Let $g_0 : \R\to \R \cup \{-\infty\}$ be such that $\sup g_0 = \gamma$. Let $(s_n)$ and $(g_n)$ be the quantities defined recursively in~\eqref{eq: SdeterministicDynamic}. For all $n \geq 1$, we have
\begin{equation}\label{eq: SdeterministicDynamic2}
\begin{split}
  &s_n = \sup\{\xi \in \R : \sup_{x \in \R} \Phi_\xi [g_{n-1}](x) \geq \gamma\},\\
  &g_n = \pi \circ \Phi_{s_n}(g_{n-1}).
\end{split}
\end{equation}
\end{prop}

\begin{proof}
We remark from the proof of Lemma \ref{lem: s} that for any $n \geq 0$, we have
\[
   \sup_{x \in \R} \Phi_\xi[g_n](x) = p_{n+1}(\xi),
\]
therefore we have $s_{n+1} = \sup\{\xi \in \R : \sup_{x \in \R} \Phi_\xi[g_{n}](x) \geq \gamma\}$ by~\eqref{eq: SdeterministicDynamic}. The formula for $g_{n+1}$ follows immediately by~\eqref{eq: SdeterministicDynamic}, and we complete the proof.
\end{proof}

We complete this section by remarking that for all $n \geq 1$, the slope of $g_n$ is bounded below by $-1$ on $\{x \in \R : g_n(x) \geq 0\}$. More precisely, we prove the following result.
\begin{lem}\label{lem: lambda}
Let $g:\R\to\R\cup\{-\infty\}$ and $\xi\in\R$. Then the function $x\mapsto \Phi_\xi[g](x)+ x$ is non-decreasing.
\end{lem}

\begin{proof}
For all $x \in \R$, we have
\begin{align*}
  \Phi_\xi[g](x) +  x &= 1-\gamma + \sup_{y \in \R} \left( g(y) - |y-x| + x \right) - \mu(\xi-x)_+\\
  &= 1 - \gamma + \sup_{y \in \R} \left( g(y) + y - 2 (y-x)_+ \right) - \mu (\xi-x)_+.
\end{align*}
As $x \mapsto - \mu(\xi-x)_+ - 2 (y-x)_+$ is non-decreasing for all $y \in \R$, we conclude that $x \mapsto \Phi_\xi[g] + x$ is also non-decreasing.
\end{proof}

\section{Traveling wave genotypic profiles}
\label{sec:travelingWave}

The main objective of the present section is to describe the concave traveling waves with compact support for the profile of our population. As we will see, these traveling waves are constructed as piecewise linear functions. Observing that there is a one-to-one map between $(v,s)$ and $(\gamma,\mu)$ will then allow us to complete the proof of Theorem~\ref{thm:main}. 

Throughout this section, $g$ will denote a concave traveling wave with speed $v$. Without loss of generality, up to translation of $g$ we can assume that
\begin{equation}\label{eq: rightedge0}
\sup\{x \in \R : g(x) > 0\} = 0.
\end{equation}
We denote by $L$ the left edge of the profile,
\begin{equation}\label{eq: L}
L:=  \inf\{x \in \R : g(x) > 0\},
\end{equation}
so that the support of $g$ is $[L,0]$, and (by concavity) we have $g(x) = -\infty$ if $x \not \in [L,0]$. In particular, we refer to $|L|$ as the width of $g$.

Recalling the definition of the functional $\Phi$ from~\eqref{eq: Phis}, we also define
\begin{equation}
  \label{eqn:defSigma}
  s := \sup\{ \xi \in \R: \sup_{x \in \R} \Phi_\xi[g](x) \geq \gamma\},
\end{equation}
which plays the role of the phenotypic threshold in a population starting with a genotypic profile $g$. In particular, $g$ being a traveling wave, we have
\[
  \forall x \in \R, \ g(x) = \pi\left(\Phi_s[g](x+v)\right).
\]
Observe that $x \mapsto \Phi_s[g](x+v)$ is a concave function $\R\to \R$ by Lemma~\ref{lem:concave}. It is in particular continuous, therefore $g(0) = g(L) = 0$, using that $\pi$ is upper semi-continuous. Therefore, the above equation can be rewritten as
\begin{equation}
  \label{eqn:travelingWaveV2}
  \forall x \in [L,0], \ g(x) = \Phi_s[g](x+v).
\end{equation}

Now we state the main result of this section, which is a more precise version of Theorem~\ref{thm:main}. We describe the fully pulled ($\gamma<\gamma_c$) and semipulled ($\gamma>\gamma_c$) traveling wave solutions to~\eqref{eq: SdeterministicDynamic} as piecewise linear functions. 

\begin{thm}
\label{x}
Let $\gamma \in (0,1)$ and $\mu>0$ with $\mu\neq 1$. 
We set $k:=\lfloor 1/\mu \rfloor$ and $K:= \lfloor 2/\mu \rfloor$. Then there exists a unique up to translation, concave, compactly supported traveling wave $g$ to the dynamic~\eqref{eq: SdeterministicDynamic}. Writing $[L,0]$ for the support of $g$, the function is differentiable everywhere on $[L,0]$ except for a finite number of points, and we have
\begin{equation}\label{eq: gint}
  g : x \in \R \mapsto \pi \left( \int_x^0 g'(y) \mathrm dy \right).
\end{equation}
Moreover, we can identify $g'$ on the interval $[L,0]$ as follows. Let $v:=v(\gamma,\mu)$ and $s$ as defined in~\eqref{eqn:defSigma}.
\begin{itemize}
\item If $\gamma < \gamma_c$, then for all $x \in [L,0]$ we have 
\begin{equation}\label{eq: slopeFully}
g'(x) = (j\mu-1), \quad 
\text{if $-jv < x < -(j-1)v$ for some $j \geq 1$,}
\end{equation}
where
\[
v=\frac{2\gamma}{k(2 - (k+1)\mu)} < s.
\]
\item If $\gamma > \gamma_c$, then for all $x \in [L,0]$, we have
\begin{equation}\label{eq: slopeSemi}
g'(x) = \begin{cases}
-1, & \text{if } x  > s - v\\
(j \mu  - 1), & \text{if } s - (j+1)v < x < s - j v \text{ for some $1 \leq j \leq K+1$} \\
1+\mu, & \text{if } x<s - (K+2)v,
\end{cases}
\end{equation}
where
\begin{equation}\label{eq: vminussigma}
v=1-\gamma \quad \text{and} \quad 
v-s = 1-(1-\gamma)\left(k\left(1-\frac{k+1}{2}\mu \right) +1 \right) =: \chi(\mu,\gamma) > 0.
\end{equation}
\end{itemize}
In particular, $\gamma\mapsto v(\mu,\gamma)$ changes monotonicity at $\gamma_{c}(\mu)$ and is maximized at that point.
\end{thm}

We first collect specific properties of the traveling wave $g$, which will be used in the proof of Theorem~\ref{x}.
We will write
\[
  m = \mathrm{argmax}(g) \in (L,0)
\]
for the leftmost point at which $g$ attains its maximum. By Lemma~\ref{lem: s} and \eqref{eqn:StravelingWave}, we remark that $g(m) = \sup g = \gamma$, so $m$ corresponds to the genotypic value shared by the largest portion of the population. Moreover, by concavity of $g$, we conclude that $g$ is increasing on $(-\infty,m]$ and non-increasing on $[m,\infty)$.

The following notation will also be useful. If there exists $x\in(L,0)$ such that $g'(x)>1$, then let $d$ denote the unique negative real number such that $g'>1$ on the interval $(L,d)$ and $g'\leq 1$ on $(d,\infty)$. Furthermore, if $g'\leq 1$ on the interval $(L,0)$, then we define $d:=L$. Now $d<0$ is well-defined because of the concavity of $g$.

\begin{lem}
\label{lem:slope}
The function $y \mapsto g(y) + y$ is non-decreasing on the interval $(-\infty,0]$, and the function $y \mapsto g(y) - y$ is non-increasing on the interval $[d,\infty)$ and increasing on $[L,d]$.
\end{lem}

\begin{proof}
As $g(x) = \pi \circ \Phi_{s}[g](x+v)$, the fact that $y \mapsto g(y) + y$ is non-decreasing on $(-\infty,0]$ follows from Lemma \ref{lem: lambda} and our convention that $g$ is non-negative on $[L,0]$ only. The second part follows from the definition of $d$.
\end{proof}

Using the definition of $\Phi_s$  immediately implies the following corollary.
\begin{cor}
\label{cor:formula}
For all $x \geq L$, we have
\begin{equation}
\label{eqn:iterate}
\Phi_s[g](x) = \begin{cases}
1-\gamma - x - \mu(s-x)_+ & \text{if } x > 0 \\
1-\gamma + g(x) - \mu(s-x)_+ & \text{if } d\leq x\leq 0 \\
1-\gamma + g(d) - (d-x) -  \mu(s - x)_+ & \text{if } x<d.
\end{cases}
\end{equation}
\end{cor}

\begin{proof}
We recall that
\[
  \Phi_s[g](x)  = 1-\gamma - \mu (s - x)_+ + \sup_{y \in \R} g(y) - |x-y|.
\]
Let $x \geq d$. By Lemma~\ref{lem:slope}, and since $g\equiv -\infty$ on $\R\setminus [L,0]$, we have that the function $y \mapsto g(y) - (x-y)$ is non-decreasing on the interval $(-\infty,x \wedge 0)$, and $y \mapsto g(y) + (x-y)$ is non-increasing on $(x,\infty)$. We conclude that
\[
  \sup_{y \in \R} g(y) - |x-y| = \begin{cases}
    g(0) - |x| &\text{if } x > 0\\
    g(x) &\text{if } x \leq 0,
  \end{cases}
\]
which completes the proof for $x\geq d$.

Now let $L\leq x<d$. Using Lemma~\ref{lem:slope} again, we see that the functions $y \mapsto g(y) - (x-y)$ and $y \mapsto g(y) + (x-y)$ are non-decreasing on the intervals $(-\infty,x]$ and $(x,d)$ respectively; and we also have that the function $y \mapsto g(y) + (x-y)$ is non-increasing on the interval $[d,\infty)$. Therefore, we obtain
\[
\sup_{y \in \R} g(y) - |x-y| = g(d) - (d-x),
\]
and the result follows.
\end{proof}

Observe that from \eqref{eqn:travelingWaveV2}, we have
\[
  g(x-v) = \pi \circ \Phi_s[g](x)
\]
for all $x \in \R$. In particular, as the right edge of the support of $g$ is $0$, we conclude that
\[
  v = \sup\{x \in \R : \Phi_s[g](x) \geq 0\}.
\]

We finish by showing that there is no concave compactly supported stationary wave, i.e.~all traveling waves have positive speed. In view of the previous corollary, we obtain a noteworthy relationship between $v$ and $s$.

\begin{lem}\label{lem: speed}
We have $v > 0$, moreover
\[
  v =
  \begin{cases}
    \frac{1 - \gamma - \mu s}{1-\mu} & \text{ if $v < s$}\\
    1-\gamma & \text{ if $v > s$}.
  \end{cases}
\]
\end{lem}

\begin{proof}
We write
\[
  p(x) = 1- \gamma + \sup_{y \in \R} g(y) - \min(\mu,1)(x-y)_+
\]
for the phenotypic profile of a population starting with a genotypic profile $g$. From \eqref{eq: SdeterministicDynamic}, we recall that $p(s) = \gamma$. Recall also the definition of the functional $\Phi_s$ from~\eqref{eq: Phis}.

We first assume that $\mu > 1$ and $s\geq 0$. Then, by the definitions of the functions $p$ and $\Phi_s[g]$, and since the support of $g$ is the interval $[L,0]$, we have $\Phi_s[g](s) = p(s) = \gamma \geq 0$. Thus,~\eqref{eqn:travelingWaveV2} implies that $g(s - v) =\gamma$. Using that $g(x) > 0$ if and only if $x \in (L,0)$, we conclude that $0<s < v$.

Next, assume that $\mu > 1$ and $s < 0$. In this case, for all $x \in (0,|s|)$, we have
\[
  \Phi_s[g](x) = 1 - \gamma + \sup_{y \in \R} g(y) - |x-y| \geq 1 - \gamma - x.
\]
Therefore $\Phi_s[g](x) > 0$ for $x > 0$ small enough, in particular $v$ is positive.

We now assume that $\mu < 1$, and first show that in this case, we have $s < \frac{1-\gamma}{\mu}$. Indeed, if this is not the case then $s \geq \frac{1-\gamma}{\mu} > 0$, and we have
\[
  \gamma = p_1(s) = 1 - \gamma + \sup_{y \in \R} g(y) - \mu (s-y)_+ =  1 - \gamma - \mu s + \sup_{y \leq 0} g(y) + \mu y < \sup_{y \leq 0} g(y),
\]
which leads to a contradiction as $\sup_{y \leq 0} g(y) = \gamma$. However, since $s < \frac{1-\gamma}{\mu}$, then by \eqref{eqn:iterate}, we have
\[
  \Phi_s[g](0) = 1-\gamma - \mu(s)_+ > 0,
\]
therefore $\Phi_s[g](x)>0$ for $x$ small enough, proving again that $v > 0$.

Now, as $v > 0$, it is defined as the unique positive root of the equation $\Phi_s[g](x) = 0$. By \eqref{eqn:iterate}, this equation can be rewritten
\[
  1 - \gamma - v - \mu(s-v)_+ = 0.
\]
Therefore, if $v > s$, we have $v = 1 - \gamma$, while if $v < s$ then $v = \frac{1-\gamma - \mu s}{1-\mu}$, completing the proof.
\end{proof}

In the rest of the section, we identify the function $g$ in terms of $v$ and $s$. More precisely, we first assume that $v < s$, which we call the \emph{fully pulled regime}, and we identify the function $g$ in this case, as well as the associated parameters $\mu$ and $\gamma$ in terms of $(v,s)$. We then work under the assumption that $v > s$, that we refer to as the \emph{semipulled regime}. We identify once again $g, \gamma$ and $\mu$ in terms of $v$ and $s$. Showing that the correspondence between $(v,s)$ and $(\gamma,\mu)$ is bijective, this is enough to prove Theorem~\ref{x}.

\begin{prop}[Characterization of fully pulled traveling waves]
\label{prop:ful}
Let $g$ be a traveling wave with speed $v$. Suppose that $v < s$. Then $g$ is given by~\eqref{eq: gint} and~\eqref{eq: slopeFully}; and we have
\[
  \mu < 1, \quad
  v = \frac{2 \gamma}{k(2 - (k+1)\mu)} \quad \text{ and } \quad s = \frac{1-\gamma - (1-\mu)v}{\mu},
\]
where $k=\lfloor1/\mu\rfloor$.
\end{prop}

\begin{proof}
The proof mainly relies on successively describing the values of $g$ on each interval $[-v,0]$, $[-2v,-v]$, ..., using~\eqref{eqn:iterate}. First observe that, since $v < s$, for all $x \in [0,v]$ we have
\[
   g(x-v) = \Phi_s[g](x) = 1 - \gamma - x - \mu (s-x) = 1 - \gamma - \mu s - (1 - \mu) x.
\]
In particular, $g$ is linear on $[-v,0]$ with slope $-(1-\mu)$ (recall that by hypothesis, $g(0)=0$). Recalling the notation $d$ from Lemma~\ref{lem:slope}, this also shows that $d\leq -v$.

We now assume that $g$ is affine on the interval $[-j v, -(j-1)v]$ with slope $-(1 - j\mu)$ for some $j \geq 1$ with $-jv\geq d$. We use again \eqref{eqn:iterate} with $x \in [-jv, -(j-1)v]$ to obtain
\[
 g(x-v) = \pi \circ \Phi_s [g] (x) = \pi\left(1 - \gamma + g(x) - \mu (s - x)\right) = \pi(1 - \gamma - \mu s + g(x) + \mu x).
\]
In other words, $g$ is affine on $[-(j+1)v,-jv] \cap [L,0]$ with slope $-(1-(j+1)\mu)$. Using finite induction, and the fact that $g$ is continuous on $[0,L]$, we conclude that $g$ is given by~\eqref{eq: gint} and~\eqref{eq: slopeFully} on $[L \vee (d-v),0]$. Recall the notation $K =\lfloor 2/\mu \rfloor$. Then the inductive argument and the fact that $-(1-K\mu)\leq 1$ also show that $d\leq L \vee -Kv$. 

Furthermore, it can also be checked that the piecewise linear function, which takes $0$ at $0$ and has slopes $-(1 - j\mu)$ on the intervals $[-j v, -(j-1)v]$, is positive at $-(K-1)v$ and non-positive at $-Kv$. Together with the fact that $d\leq L \vee -Kv$, this shows that $d=L$ and that $g$ is given by~\eqref{eq: gint} and~\eqref{eq: slopeFully} on the interval $[L,0]$.

We recall from Lemma~\ref{lem: speed} that, since $v < s$, we have
\begin{equation}
  \label{eqn:firstRelationshipFully}
  v = \frac{1-\gamma - \mu s}{1-\mu},
\end{equation}
giving a first equation between the parameters $(v,s)$ and $(\mu,\gamma)$. To find the other one, we recall that
$\sup_{y \in \R} g(y) = g(m) = \gamma$, which, given the formulas~\eqref{eq: gint} and~\eqref{eq: slopeFully}, provides a second relationship between the parameters.

More precisely, let us recall that $k= \floor{1/\mu}$, thus
\[
  k \mu \leq 1 < (k+1) \mu.
\]
As a result, the slope of $g$ on $[-kv,-(k-1)v]$ is $-(1 - k\mu) \leq 0$, while the slope of $g$ on $[-(k+1)v,-kv]$ is $-(1 -(k+1)\mu) > 0$. Since $g$ is concave, we conclude that $m = \mathrm{argmax}(g) = -kv$, and that $g(-kv) = \gamma$. In particular, this condition implies that $k > 0$, because otherwise we would have $\gamma = g(0)= 0$. Hence, we conclude $\mu < 1$. Simple computation from~\eqref{eq: gint} and~\eqref{eq: slopeFully} then yields 
\[
  \gamma = g(-kv) = \sum_{j=1}^{k} v (1 - j \mu) = v \left(k - \frac{k(k+1)}{2}\mu\right),
\]
therefore we obtain the second relationship
\begin{equation}
  \label{eqn:secondRelationshipFully}
  \gamma = \frac{kv}{2}(2 - (k+1)\mu).
\end{equation}

From \eqref{eqn:firstRelationshipFully} and \eqref{eqn:secondRelationshipFully}, we deduce the conditions
\begin{equation}
  \label{eqn:formula}
  \begin{cases}
    & v = \frac{2 \gamma}{k(2 - (k+1)\mu)}\\
    & s = \frac{1 - \gamma - (1-\mu) v}{\mu},
  \end{cases}
\end{equation}
completing the proof.
\end{proof}

We now turn to the study of the traveling waves associated to the semipulled regime.
\begin{prop}[Characterization of the semipulled traveling waves]
\label{prop:sem}
Let $g$ be a traveling wave with speed $v$. Suppose that $v > s$. Then, $g$ is given by \eqref{eq: gint} and~\eqref{eq: slopeSemi}, and we have
\[
  v = 1 - \gamma \quad \text{ and } \quad  s = \frac{(k+1)(1-\gamma)}{2}(2 -k \mu) - \gamma.
\]
\end{prop}

\begin{proof}
We use once again \eqref{eqn:iterate}, first to give the value of $g$ on the interval $[s - v, 0]$, then on each interval $[s - (j+1)v, s - j v]$ by recursion. As a first step, we show that
\begin{equation}
  \label{eqn:aim1}
  g(x) = - x \quad \text{ for all $x \in [s-v, 0]$.}
\end{equation}

Using \eqref{eqn:iterate}, we observe that for $x \in [s_+,v]$, we have
\[
  g(x-v) = \Phi_s[g](x) = 1 - \gamma - x.
\]
Thus, as $v = 1 - \gamma$ by Lemma~\ref{lem: speed}, we have $g(x) = -x$ for all $x \in [s_+-v,0]$. Therefore, if $s \geq 0$, then the proof of~\eqref{eqn:aim1} is now complete.

If $s < 0$, we prove recursively that for all $i\geq 1$ and $x \in [s_+ - i v, s_+ - (i-1)v] \cap [s-v,0]$, we have $g(x) = -x$. This result being proved for $i=1$, we assume it to hold for some $i \geq 1$. Then by~\eqref{eqn:iterate}, for all $x \in [s_+- iv, s_+- (i-1)v]$, if $x > s - v$, we have
\[
  g(x) = \Phi_s[g](x + v) = 1 - \gamma + g(x+v) = 1 - \gamma - (x + v) = -x,
\]
proving~\eqref{eqn:aim1} by induction. The argument also shows that $d\leq s-v$ (recall the notation $d$ from Lemma~\ref{lem:slope}).

We now turn to the description of the profile of $g$ on the interval $[s - 2v, s - v]$. For all $x \in [s - v,s]$, we have
\begin{align*}
  g(x-v) &= \Phi_s[g](x) = \begin{cases}
    1 - \gamma + g(x) - \mu(s - x)_+ & \text{ if } x \leq 0\\
    1 - \gamma -x - \mu(s - x)_+ &\text{ if } x > 0
  \end{cases}\\
  &= 1 - \gamma - x - \mu(s - x),
\end{align*}
showing that $g$ has slope $-(1-\mu)$ on $[s - 2v, s - v]$. 

Using~\eqref{eqn:iterate} and the same method as in the proof of Proposition~\ref{prop:ful}, we can again prove by induction that $g$ has slope $-(1-j\mu)$ on $[s - (j+1)v,s-jv]\cap [L,0]$, for all $j\leq K+1$ with $K =\lfloor 2/\mu \rfloor$. If $L<s - (K+2)v$, then we have $d=s-(K+1)v$, and by the third line of~\eqref{eqn:iterate}, the slope of $g$ will be $1+\mu$ on the interval $[L,s - (K+2)v]$. We therefore conclude that the traveling wave $g$ is given by~\eqref{eq: gint} and~\eqref{eq: slopeSemi}.

Finally, we determine the relationship between $(v,s)$ and $(\mu,\gamma)$. From \eqref{eq: gint} and~\eqref{eq: slopeSemi} we see that the slope of $g$ changes sign at $s - (k+1)v$, and so we have $m = \mathrm{argmax}(g) = s-(k+1)v$ and
\[
  \gamma = \sup g = g(s-(k+1)v).
\]
Hence, a simple computation leads to
\begin{align}
  \gamma = v - s + \sum_{j=1}^{k} v (1 -j\mu) 
  = \frac{(k+1)v}{2}(2 - k\mu) - s. \label{eqn:secondEqualitySemi}
\end{align}
This last equality, together with $v = 1 - \gamma$ allows us to write
\[
  s = \frac{(k+1)(1-\gamma)}{2}(2 -k \mu) - \gamma,
\]
completing the proof.
\end{proof}

We now have all the preliminary results needed to prove Theorem~\ref{x}.

\begin{proof}[Proof of Theorem~\ref{x}]
Let $\gamma \in (0,1)$ and $\mu > 0$. We set
\[
  v_\mathrm{f} = \frac{2 \gamma}{k(2 - (k+1)\mu)} \quad \text{and} \quad v_\mathrm{s} = 1 - \gamma,
\]
as well as
\[
  s_\mathrm{f} =  \frac{1-\gamma - (1-\mu)v_\mathrm{f}}{\mu} \quad \text{and} \quad s_\mathrm{s} = \frac{(k+1)(1-\gamma)}{2}(2 - k\mu) - \gamma,
\]
which are respectively the speed and phenotypic threshold of traveling waves in the fully pulled and semipulled regimes. We observe that
\begin{align*}
   v_\mathrm{f} < s_\mathrm{f}
   &\iff \mu v_\mathrm{f} < 1 - \gamma - (1-\mu)v_\mathrm{f} \iff v_\mathrm{f} < 1-\gamma\\
   &\iff \frac{2\gamma}{k(2-(k+1)\mu)} < 1-\gamma \iff \gamma \left( \frac{2}{k(2-(k+1)\mu)} +1 \right)< 1 \\
   &\iff \frac{\gamma}{\gamma_c(\mu)} < 1,
\end{align*}
and we note that these inequalities hold only if $\mu < 1$, in which case $\gamma_c(\mu) > 0$. Similarly, we have
\begin{align*}
  v_\mathrm{s} > s_\mathrm{s} &\iff 1 - \gamma > \frac{(k+1)(1-\gamma)}{2}(2 - k\mu) -\gamma \iff \gamma (k+1)(2 - k\mu) > (k+1)(2 - k \mu) - 2\\
  &\iff \gamma > \frac{(k+1)(2 - k \mu) - 2}{(k+1)(2 - k\mu)} = \frac{k(2-k\mu) - k \mu}{(k+1)(2-k\mu)} = \gamma_c(\mu).
\end{align*}
Moreover these inequalities hold for any $\mu \neq 1$.

If $\gamma < \gamma_c(\mu)$ then we have $\mu<1$, and so we also have $v_\mathrm{f}  < s_\mathrm{f}$. Therefore, the traveling wave $g$ described in Proposition \ref{prop:ful} is a solution to \eqref{eqn:StravelingWave} with speed $v = v_\mathrm{f}$. On the other hand, $\gamma < \gamma_c(\mu)$ also implies $v_\mathrm{s} < s_\mathrm{s}$; therefore there is no solution of the form described by Proposition \ref{prop:sem}. We conclude that if $\gamma < \gamma_c(\mu)$ then the function given by~\eqref{eq: gint} and~\eqref{eq: slopeFully} is the unique traveling wave solution to the dynamic~\eqref{eq: SdeterministicDynamic} with $v=v_\mathrm{f}< s$.

Similarly, if $\gamma > \gamma_c(\mu)$, we have $v_\mathrm{s} > s_\mathrm{s}$, thus the traveling wave $g$ described in Proposition~\ref{prop:sem} is a solution to~\eqref{eqn:StravelingWave} with speed $v=v_\mathrm{s}$. Furthermore, if $\gamma>\gamma_c$ and $\mu<1$, then $v_\mathrm{f}>s_\mathrm{f}$, and when $\mu>1$, then we must have $v>s$ by Proposition~\ref{prop:ful}. Thus, Proposition~\ref{prop:ful} does not provide any solution in this case. We conclude that, if $\gamma > \gamma_c(\mu)$, then the unique traveling wave solution to~\eqref{eqn:StravelingWave} is given by~\eqref{eq: gint} and~\eqref{eq: slopeSemi} with speed $v=v_\mathrm{s}>s_s=s$, and $v-s=\chi(\mu,\gamma)$ from~\eqref{eq: vminussigma}.
\end{proof}

\begin{proof}[Proof of Theorem~\ref{thm:main}]
This result is an immediate consequence of the more precise Theorem~\ref{x}.
\end{proof}

\begin{proof}[Proof of Theorem~\ref{thm:reprod}]
Recall that
\[
r: x\in\R \mapsto \pi\left[1 - \gamma + \sup_{y \in \R} \left( g(y) - |x-y|\right)\right]
\]
denotes the reproduction profile when we start from the genotypic profile $g$. For any $g:\R\to\R\cup\{-\infty\}$, let us define the functional $\Psi$ by
\[
\Psi[g]: x \mapsto 1 - \gamma + \sup_{y \in \R} \left( g(y) - |x-y|\right),
\]
so that $r(x) = \pi(\Psi[g](x))$. Now notice that 
\begin{equation}\label{eq: PsiPhi}
   \Psi[g](x) = \Phi_s[g](x) + \mu(s-x)_+ 
\end{equation}
for all $x\in\R$, and since the function $\Phi_s[g]$ is continuous, so is $\Psi[g]$.

If $\gamma<\gamma_c$, then we have $v<s$ by Theorem~\ref{x}. Then for any $x$ in the support of $g_1$, we have $x\leq v < s$, and thus, by~\eqref{eq: PsiPhi} we find 
\[
\Psi[g](x) = g_1(x) + \mu(s-x) > g_1(x),
\]
which concludes the first part of the theorem. 

If $\gamma<\gamma_c$, then we have $v<s$ by Theorem~\ref{x}. Thus, since $\Phi_s[g](v) = 0$,~\eqref{eq: PsiPhi} implies $\Psi[g](v) > 0$. Then since $\Psi[g]$ is continuous, there exists $x\in(v,s)$ such that $\Psi[g](x)>0$ and therefore $r(x)>0$, which concludes the first part of the theorem.

If $\gamma>\gamma_c$, then $v>s$ by Theorem~\ref{x}. Therefore, using~\eqref{eq: PsiPhi} and the fact that $\Phi_s[g](x)\geq 0$ for all $x\in[s,v]$, we obtain
\[
0\leq g_1(x)= \Phi_s[g](x) = \Psi[g](x) = r(x), 
\]
for all $x\in[s,v]$, which finishes the proof.
\end{proof}

\section{Ancestral structure of the model}\label{sect: ancestry}
We now turn to the study of the ancestral structure of the population,  leading to a proof of Theorem~\ref{thm:ancestralLineages}. The main question of the section is the following: If we sample $k$ individuals at a given time horizon, what can be inferred about the genealogical tree formed by tracing the $k$ ancestral lineages of those individuals? As we will see, our hydrodynamic limit already offers insights into the behavior of a single ancestral lineage ($k=1$), which will be discussed in Section \ref{sect:1-lineage}.

For multiple lineages ($k>1$), the coalescence times are primarily influenced by the stochastic fluctuations in the system. To gain further insight, we examine an alternative individual-based model with noisy selection. Although the specifics of this model differ, this model is fully integrable and shares certain universal properties with our original selection model as the population size approaches infinity. This comparative approach will allow us to derive the ansatz for the effective population size $N_e$ introduced in Section \ref{sect:effecti-pop}, and will be the focus of Section \ref{sect:k-lineages}.

\subsection{Parental lineage}\label{sect:1-lineage}
 Let $g$ be a traveling wave of speed $v$, and recall that the ancestral map describes the most likely location (relative to the rightmost genotype) of the parent  of a uniformly chosen genotype around location $x$ (in the log scale) $1$ generation in the past:
\begin{equation*}
	A(x) :=  \mathrm{argmax}_{y\in\R} \{g(y) - |x+v-y|\},
\end{equation*}
where $\mathrm{argmax}_{y \in \R} \{g(y)-|x+v-y|\}$ is defined as the smallest $y$ at which the maximum is attained, and
\begin{equation*}
	A^+(x) :=  \mathrm{Argmax}_{y\in\R} \{g(y) - |x+v-y|\}
\end{equation*}
is the largest $y$ at which the maximum is attained. As the main step towards the proof of Theorem~\ref{thm:ancestralLineages}, we show the following result.
\begin{lem}
\label{lem: Ax}
Let $g$ be a traveling wave of speed $v$ with phenotypic threshold $s$.  Then for all $x \in [L,0]$ the following holds.
\begin{itemize}
\item 
If $\gamma<\gamma_c$ then
\begin{equation*}
  A(x) = A^+(x) = \min(x+v,0).
\end{equation*}
\item 
If $\gamma>\gamma_c$, then
\begin{equation*}
A(x) = \min(s-v, \max(x+v,d)) \quad \text{and} \quad A^+(x) = \min(0, \max(x+v,d)).
\end{equation*}
\end{itemize}
\end{lem}

\begin{proof}
Let $x\in[L,0]$ be the location of a genotype relative to the rightmost genotype at generation $1$, that is, relative to $v$. Let $c\in[L+v,v]$ be the actual location of this genotype, so that $c=x+v$. 

If $\gamma < \gamma_c$, then Theorem~\ref{x} shows that $g'\in (-1,1]$ (we can only have $g'=1$ on the interval $[L,-(K-1)v]$ when $2/\mu$ is an integer). Therefore, it is not hard to see that the maximum of $y \mapsto g(y) - |c-y|$ is unique for all $c\in[L+v,v]$, and it is attained at $c$ for $c\in[L+v,0]$, and at $0$ for $c\in[0,v]$. With the change of variables $x=c-v$ we see that $A(x) = A^+(x) = \min(x+v,0)$. 

We now turn to the case $\gamma>\gamma_c$. We first take $c\in[s-v,v]$. We use Theorem~\ref{x} again to determine where the maximum of the function $y \mapsto g(y) - |c-y|$ is attained. Using the slopes given by the theorem, in particular noting that $g'=-1$ on the interval $(s-v,0)$, we find that the maximum is attained everywhere on the interval $[s-v, 0 \wedge c ]$. 

Similarly, using Theorem~\ref{x} one can see that the maximum of the function $y \mapsto g(y) - |c-y|$ is attained at $c$, if $c\in[(L+v)\vee d,s-v]$, and it is attained at $d$, if $c\in[L+v,d]$. 

With the change of variables $x=c-v$ and using the definitions of the functions $A$ and $A^+$, we obtain
\begin{equation*}
A(x) =
\begin{cases}
s - v, & \text{if } x\in [s-2v,0] \\
x+v, & \text{if } x\in [L\vee(d-v),s-2v] \\
d, & \text{if } x\in[L,d-v],
\end{cases}
\end{equation*}
and
\begin{equation*}
A^+(x) = 
\begin{cases}
\min(x+v,0), & \text{if } x\in [s-2v,0] \\
A(x), & \text{if } x\in [L,s-2v].
\end{cases}
\end{equation*}
The two relations above imply the second part of the lemma.
\end{proof}

\begin{proof}[Proof of Theorem~\ref{thm:ancestralLineages}]
Lemma~\ref{lem: Ax} implies $A(x) = A^+(x) = 0$ for all $x\in[-v,0]$, if $\gamma<\gamma_c$; and also $A(x) = s-v$ and $A^+(x) = \min(x+v,0)$ for all $x\in[s-2v]$, if $\gamma>\gamma_c$.

Lemma~\ref{lem: Ax} also implies that for any $x\in[L,0]$, both $A^j(x)$ and $(A^+)^j(x)$ increase by at least $v$ until they reach $0$ or $s-v$. Since $L<\infty$ and $v>0$ in both regimes ($\gamma<\gamma_c$ and $\gamma>\gamma_c$), the statements of the theorem about $A^j(x)$ and $(A^+)^j(x)$ follow.	
\end{proof}

\medskip

\subsection{Genealogical structure.}\label{sect:k-lineages}

We are now interested in the genealogical tree generated by tracing backward in time $k$
distinct ancestral lineages. Unfortunately, the genealogical structure is inherently stochastic, and random coalescence times can not be read from the hydrodynamic limit. Thus, we need to go back to the stochastic model and understand the fluctuations of the system. 

Rather than analyzing the original model directly, we consider a fully integrable variation. As we will demonstrate, this integrable model retains many key properties, suggesting that both the original and modified versions belong to the same universality class and share an identical genealogical structure. Given that the ancestral structure of the integrable model is known \cite{schertzer2023relative}, we can use a comparative approach to derive the ansatz for $N_e$ presented in the main text (see equations~\eqref{eq:T2-strong} and~\eqref{eq:T2-weak}. The following paragraphs will elaborate on this. 
\medskip

{\it The noisy exponential $K$-BRW model.} We consider an extension of the exponential $K$-branching random walk ($K$-BRW) model of Brunet and Derrida~\cite{Brunet2007}. It was first introduced in \cite{Cortines2017} by Cortines and Mallein, and further analyzed in Schertzer and Wences \cite{schertzer2023relative}.

As in the original model, at every generation, the population is made of $K=N^\gamma$ individuals and evolves in two steps at every generation.

\medskip

{\it Reproduction.} An individual with genetic value $x$ produces an infinite number of offspring whose genotypes are distributed according to an independent exponential Poisson point process (PPP) with intensity measure $e^{-(y-x)}dy$. Note that the exponential PPP is shifted in such a way that the distribution is centered at $x$. In particular, there are only finitely many offspring to the right of $x$, and infinitely ``unfit" children to its left.

\medskip

{\it Selection.} After reproduction,  infinitely many children are present. 
We then select the $N^\gamma$ individuals using a sampling scheme interpolating between truncation selection and Gibbs sampling as follows. 
Let $\mu>0$. Because the intensity measure vanishes exponentially fast at $\infty$, it can be shown that children after reproduction can be ranked in decreasing order. We first select the $N$ rightmost genotypes (truncation selection), and then sample $N^\gamma$ individuals without replacement according to the sampling weights $e^{\mu x}$ (Gibbs sampling).  
When $\mu=\infty$, this amounts to selecting the $N^\gamma$ rightmost individuals. For $\mu=0$, $N^{\gamma}$ individuals are selected uniformly at random 
from the $N$ rightmost children. Thus, 
as in the original model, the $\mu$ parameter  also captures the level of noise in the selection scheme.

\bigskip
\noindent
{\bf Universal traveling wave.} 
The noisy exponential $K$-BRW and our original model share the same phenomenology  summarized in Table \ref{tab:exponential}. For the sake of presentation, we will restrict ourselves to the ``shape" and speed of the traveling wave solution presented in the next table. 

\begin{table}[h!]
\caption{Universal phenomenological picture of the fully pulled and semipulled traveling waves.}
\begin{center}	
\begin{tabular}{|m{3.5cm} | m{6cm} | m{6cm}|}
\hline
& \textbf{Fully pulled wave} & \textbf{Semipulled wave} \\
\hline
& $\gamma < \gamma_c(\mu)$ & $\gamma > \gamma_c(\mu)$   \\
\hline
\textbf{Stationary profile} & & \\
\hline
Speed & non-decreasing function of $\gamma$ & non-increasing function of $\gamma$ (flat in the noisy exponential $K$-BRW) \\
\hline
Slope of $g$ at the front & $\mu-1$ & $-1$ and then $\mu-1$ after some $-\chi>0$ \\
\hline
\end{tabular}	
\end{center}
\label{tab:exponential}
\end{table}

In the following, we justify the content of Table~\ref{tab:exponential} for the noisy exponential $K$-BRW. The justification for our original model was presented in Table 
\ref{tab:phenomenology}.
Let us consider an initial configuration of particles $(x^i_0)_{i=1}^{N^\gamma}$ with a limiting log-profile $g_0$, so that the number of particles around a location $x\log N$ is roughly $N^{g_0(x)}dx$. Formally, we assume the existence of a function $g_0$ valued in $\R_+\cup \{-\infty\}$ such that 
for every $a<b$, as $N\to\infty$,
$$
\frac{\log(\#\{i : x_0^i \in (a\log(N), b\log(N)) \})}{\log(N)}  \to \max_{(a,b)} g_0 \ \ \mbox{in probability.} 
$$

The key observation about the noisy exponential $K$-BRW is that the superposition of shifted exponential PPPs is again a shifted exponential PPP.  More precisely, 
if the ${\cal P}_{i}$'s are independent exponential PPPs with respective intensity $e^{-(x-x_0^i)}dx$ (describing the position of the offspring after the reproduction step), then
$$
\sum_{i} {\cal P}_i \ = \ {\cal P} \ \ \ \mbox{in law,}
$$
where ${\cal P}$ is again a shifted exponential PPP with intensity $e^{-(y-X_{eq}^N)}dy$ with a shift
$$
X_{eq}^N \equiv X_{eq}^N(x_0^i) \ := \ \log\left( \sum_{i=1}^{N^{\gamma}} e^{x_0^i}\right).
$$
 We emphasize that this simple but crucial  observation by Brunet and Derrida \cite{Brunet2007} makes the model fully integrable. We now make use of this fact to compute the genotypic profile after one generation. 

\bigskip
\noindent
{\it 
Reproduction profile.}
Let us first consider the individuals $(r_0^i)_{i=0}^\infty$ after reproduction. By the previous observation, 
\begin{eqnarray}
\frac{\log(\mathbb{E}[\# \{i :r_0^i \in X_{eq}^N + (a\log(N),b\log(N)) \}])}{\log(N)} & = &  \frac{\log\int_{a\log(N)}^{b\log(N)} e^{-x} dx}{\log(N)} \nonumber \\
 & \to  & \max_{x\in(a,b) } (- x), \label{eq:expected} 
\end{eqnarray}
as $N\to\infty$.
By a second moment argument, one can easily prove that the expectation can be removed inside the log, which yields
\begin{eqnarray}
\frac{\log(\# \{i :r_0^i \in X_{eq}^N + (a\log(N),b\log(N)) \})}{\log(N)} & \to  & \max_{(a,b)} R, \ \mbox{where} \ R(x):= \pi(- x),
\end{eqnarray}
where the convergence is meant in probability and
the projector $\pi$ has the effect of setting the population to $0$ when the ``expected stochastic exponent" in (\ref{eq:expected}) takes negative values (as highlighted in the main text). 

For any $a\in\R$, define the shift operator $\theta_a$ by
$$
\forall x>0, \ \theta_{a}f(x) = f(x-a)
$$
and 
$$
\hat X_{eq} = \frac{X_{eq}}{\log(N)}.
$$
The previous result implies that the set of individuals after reproduction has a limiting log-profile given by $\theta_{\hat X_{eq}} R(x)$.

\bigskip

\noindent
{\it Truncation.} We now consider the system of particles after only retaining the $N$ rightmost individuals. This leads to a truncation profile given by $\theta_{\hat X_{eq}}T(x)$ where
\begin{eqnarray}
\label{eq:truncation-profile}
    T(x) = \left\{\begin{array}{cc} -x& \ \mbox{if }  
 \ \ -1<x < 0 \\ -\infty & \mbox{otherwise.}
    \end{array} \right.
\end{eqnarray} 
 In words, the log-profile after truncation is obtained by cutting the reproduction profile to the left of $(\theta_{\hat X_{eq}} R)^{-1}(1)=\hat X_{eq}-1$, so that the  $\asymp N$ rightmost particles remain.

\bigskip
\noindent
{\it Gibbs selection.} 
Let us consider the set of individuals present after the truncation step. Let ${\cal T}=(z_0^i)_{i=1}^{N}$ denote their positions relative to $\hat X_{eq}$. We now select $N^\gamma$ particles without replacement according to the sampling weights $e^{\mu z_0^i}$. 
Let $z\in{\cal T}$ be an individual at position $u\log(N)$ ($u\in(-1,0)$). 
For the sake of simplicity, let us first assume that individuals are sampled with replacement and let
$$
q_N(u) := N^{\gamma} \times \frac{e^{\mu u\log(N)}}{\sum_{i=1}^N e^{\mu z_0^i}}
$$
be the expected  number of times our focal individual is selected. From~\eqref{eq:truncation-profile}, an easy computation shows that 
$$
\frac{\log(\sum_{i=1}^N e^{\mu z_0^i})}{\log(N)} \approx {1-\mu},
$$
so that 
$$
\frac{\log(q_N(u))}{\log(N)} \approx \gamma-1+\mu(1+u).
$$
As a conclusion, if we sample with replacement then
\begin{itemize}
\item If $\gamma-1+\mu(1+u)>0$,
the individual is sampled infinitely many times as $N\to\infty$.
\item If $\gamma-1+\mu(1+u)<0$,
the probability of sampling the individual goes to $0$ and is 
$\asymp  N^{\gamma-1+\mu(1+u)}$.
\end{itemize}
With a little bit of extra work, we can then deduce that if we now sample without replacement (as we should), then the following dichotomy holds
\begin{itemize}
\item If $\gamma-1+\mu(1+u)>0$
the probability of sampling the individual goes to $1$.
\item If $\gamma-1+\mu(1+u)<0$
the probability of sampling the individual is $\asymp N^{\gamma-1+\mu(1+u)}$.
\end{itemize}
We can then deduce that the log-profile of genotypes after one generation is $\theta_{\hat X_{eq}(x_0)} G$,
where
$$
G(x) = 
    \left\{\begin{array}{cc} \pi( -x + (\gamma-1+\mu(1+x))_-)& \ \mbox{if \ \ } -1<x < 0 \\ -\infty & \mbox{otherwise.}
    \end{array} \right.
$$

\bigskip

To summarize the previous heuristics, one striking feature of the noisy exponential $K$-BRW is that the wave reaches ``stationarity" after only $1$ generation and the only effect of the initial configuration is in the shift $\theta_{\hat X_{eq}(x_0)}$.
It then follows that $G$
is a traveling wave solution with a speed given by the limit of $\hat X_{eq}(x_0)$ ,
where $(x_0^i)_{i=1}^{N^\gamma}$ is a configuration with a limiting log-profile $G$.

\medskip

As in our original model, we can now distinguish between two regimes from the explicit description of $G$.
Define 
$$
\chi(\mu,\gamma) \equiv \chi := 
1 - \frac{1-\gamma}{\mu} 
$$
and define 
$$
\gamma_{c}(\mu) := 1-\mu
$$
so that $\chi >0$ if and only if $\gamma>\gamma_c$. 

\bigskip

\noindent
{\it Semipulled regime ($\gamma > \gamma_c$).}
 Then $G$ is obtained by concatenating continuously two linear functions with respective slopes $-1$ and $-(1-\mu)$ at $-\chi$. That is
 $$
 G(x) \ = \ 
 \left\{\begin{array}{cc} 
 -x & \ \mbox{if $x\in(-\chi,0)$} \\
 \chi  -(1-\mu)(x-\chi) \ & \ 
 \mbox{if $x\in(-1,-\chi)$} \\
 -\infty & \mbox{otherwise.}
 \end{array} \right.
 $$
From the previous computations, we see that 
 the slope at the tip is $-1$ and then $-(1-\mu)$. The change of slope occurs at $-\chi$. 
 
Let us now consider the speed of the wave and its monotonicity in $\gamma$. 
Recall that the speed is given by
\begin{equation}
\label{eq:hatX}
\hat X_{eq}(x_0^i) \ = \frac{1}{\log(N)} \ \log(\sum_{i=0}^{N^\gamma} e^{x_0^i} ),
\end{equation}
where 
$(x_0^i)$ has limiting log-profile $G$. The expression of $G$ and a direct computation yields that  the speed
is given by
 \begin{equation}
 \label{eq:speed-semi}
 v \approx \frac{\log(\chi(\mu,\gamma)\log N)}{\log(N)}  \ \to \  0, \ 
 \end{equation} 
as $N\to\infty$, so that 
in the semipulled regime, the wave is static in the natural scaling of the system (that is, in $\log(N)$ units).

\bigskip

\noindent
{\it Fully pulled regime ($\gamma < \gamma_c$).}
Define $\alpha$ through the relation
$$
\alpha+ (\gamma-1 +\mu(1-\alpha) )=0 \ \ \Longleftrightarrow \ \  \alpha = 1-\frac{\gamma}{1-\mu} \in (0,1).
$$
Then 
$$
 G(x) \ = \ 
 \left\{\begin{array}{cc} 
 -(1-\mu)(x+\alpha)  & \ \mbox{if $x\in(-1,-\alpha)$}  \\
 -\infty & \mbox{otherwise,}
 \end{array} \right.
 $$
 and a direct computation from (\ref{eq:hatX}) shows that 
 \begin{equation}
 \label{eq:speed-fully}
v(\mu,\gamma) \approx -\alpha =  \frac{\gamma}{1-\mu}-1<0.      
 \end{equation}
 It follows that $v$ is now increasing in $\gamma$. 

 \bigskip

 Finally, putting all the previous results together yields Table \ref{tab:exponential}.

\medskip

\noindent
{\bf Genealogies.} The noisy exponential $K$-BRW preserves the properties of the original model summarized in Table \ref{tab:exponential}. Those properties were derived (analogously to the original model) by looking at the limiting log-profile after one generation.
Since all other properties listed in Table \ref{tab:phenomenology} 
(ancestry, selection, etc.) follow from this analytical approach, it is not hard to extend the previous computations and show that actually all the properties listed in Table \ref{tab:phenomenology} also hold for the noisy exponential $K$-BRW.

This hints at the fact that the two models fall in the same universality class, so that the same genealogical structure should emerge in the infinite population limit.
Let us now recall one of the main results for the noisy exponential $K$-BRW derived in \cite{schertzer2023relative}. See Theorem 2.7 in \cite{schertzer2023relative} for a more precise statement.

For a population of size $N$, let $\Pi^{N}_k$ be the random genealogy obtained by sampling $k$ individuals at a given time horizon and tracing their ancestral lineages backward in time.  
In \cite{schertzer2023relative}, it is shown that for a fixed value of $\mu\in(0,1)$,  
\begin{enumerate}
\item If $\gamma<\gamma_{c}(\mu)$, then $\Pi^{N}_k$ converges to the (discrete) Poisson–Dirichlet coalescent with parameter $(1-\mu,0)$. See \cite{schertzer2023relative} for a definition. In particular, lineages coalesce in finite time for $k=2$, and the effective population size is
given by
\begin{equation}
N_e \equiv \E{T_2^N}\approx\frac{1}{\mu}.
\label{eq:Ne-below-gammac}
\end{equation}
\item If $\gamma>\gamma_c(\mu)$, then the coalescence time between two lineages goes to $\infty$ and we need to accelerate time by $\chi \log(N)$ in order to see an interesting picture emerging. After this proper time rescaling, the tree $\Pi^N_k$ converges to the Bolthausen-Sznitman coalescent \cite{Pitman1999}. In particular, \cite{schertzer2023relative} proved that 
\begin{equation}
N_e=\E{T_2^N}\approx \chi(\mu,\gamma) \log(N).
\label{eq:Ne-above-gammac}
\end{equation}
\end{enumerate}

This is the ansatz used in Section \ref{sect:effecti-pop} for the model with phenotypic noise which is in good accordance with our numerical simulations. 
See Fig \ref{fig: rateOfCv} in the main text. 

\section{Corrections to the limit theorems}
\label{sec:corrections}
    
 The question of the deviation of finite size models from their  deterministic approximations is of fundamental importance since those limits involve log transforms and change of scales in $\log(N)$. As expected, our numerical simulations show measurable errors from their infinite population prediction. See again Figure \ref{fig: rateOfCv} in the main text. 

 Before explaining the origin of those deviations, we note that despite the fact that our hydrodynamic results only provide rough approximates on the precise values for the rate of adaptation and  the effective population size (even for large $N$), our numerical simulations show that our limit theorems still capture the main qualitative behavior of the system. In particular, we can predict the existence of a phase transition between the selection of the luckiest and fittest regimes (change of monotonicity in the evolution speed $\gamma\mapsto v(\mu,\gamma)$ and in the effective population size $\mu\mapsto N_{e}(\mu,\gamma)$). Furthermore, the critical value $\gamma_{c}(\mu)$
is well predicted by our limit theorems. See the right panel of Figure~\ref{fig: rateOfCvSM}. 

\medskip

{\it Order of corrections.}
Since the noisy exponential $K$-BRW (with $K=N^\gamma$) is fully integrable, it can inform us on the order of the finite-size population deviations from the limit theorems. 
We have the following picture:
\begin{enumerate}
\item 
for $\gamma<\gamma_{c}(\mu)$, it is shown in \cite{schertzer2023relative} that
\begin{equation}
\label{eq:speed_correction1}
v_N \approx -(1-\frac{\gamma}{1-\mu}) + \frac{\E{\log(Y_\mu)}}{\log(N)} + o\left(\frac{1}{\log(N)}\right),
\end{equation}
where $Y_{\mu}$ is a positive $(1-\mu)$-stable random variable whose Laplace transform is given by
\begin{equation}\label{eq:speed_correction2}
\E{e^{-\lambda Y_{\mu}}}) = \exp(-\Gamma(\mu)\lambda^{1-\mu});
\end{equation}
\item for $\gamma>\gamma_{c}(\mu)$, we showed in \ref{eq:speed-semi} that
$$
v_N \approx   \frac{\log(\chi(\mu,\gamma) \log(N))}{\log(N)}.
$$
\end{enumerate}
As claimed in the main text, we then have a correction of order $1/\log(N)$ when $\gamma<\gamma_{c}$, or a correction of
$\log(\chi(\mu,\gamma) \log(N))/\log(N)$ when $\gamma>\gamma_{c}$. 
Note that the latter correction explodes when $\gamma\to \gamma_c(\mu)$, indicating that a different theory is needed to grasp the behavior of the system at and near criticality.

In addition, we plotted the estimation of the speed for different values of $N$ in the noisy exponential $K$-BRW. See the left panel of Figure~\ref{fig: rateOfCvSM}. We see that the higher order correction terms make the convergence to the theoretical rescaled speed extremely slow. This is particularly true in the semipulled regime ($\gamma>\gamma_{c}(\mu)$) because of the correction $\log(\chi \log(N))/\log(N)$. 
A similar pattern is observed in our original model (see the left panel of Fig.~\ref{fig: rateOfCv}). Extending our methods to quantify the error terms (as in the noisy exponential $K$-BRW) remains an important but presumably difficult challenge. 

\begin{figure}[H]
\begin{center}
    \begin{tabular}{cc}    
    \includegraphics[width=0.4\textwidth]{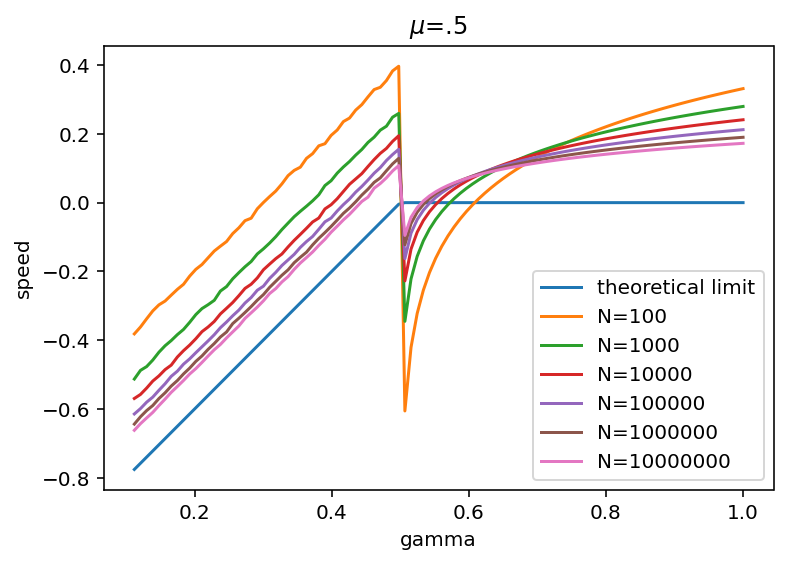}
    &
    
    \includegraphics[width=0.4\textwidth]{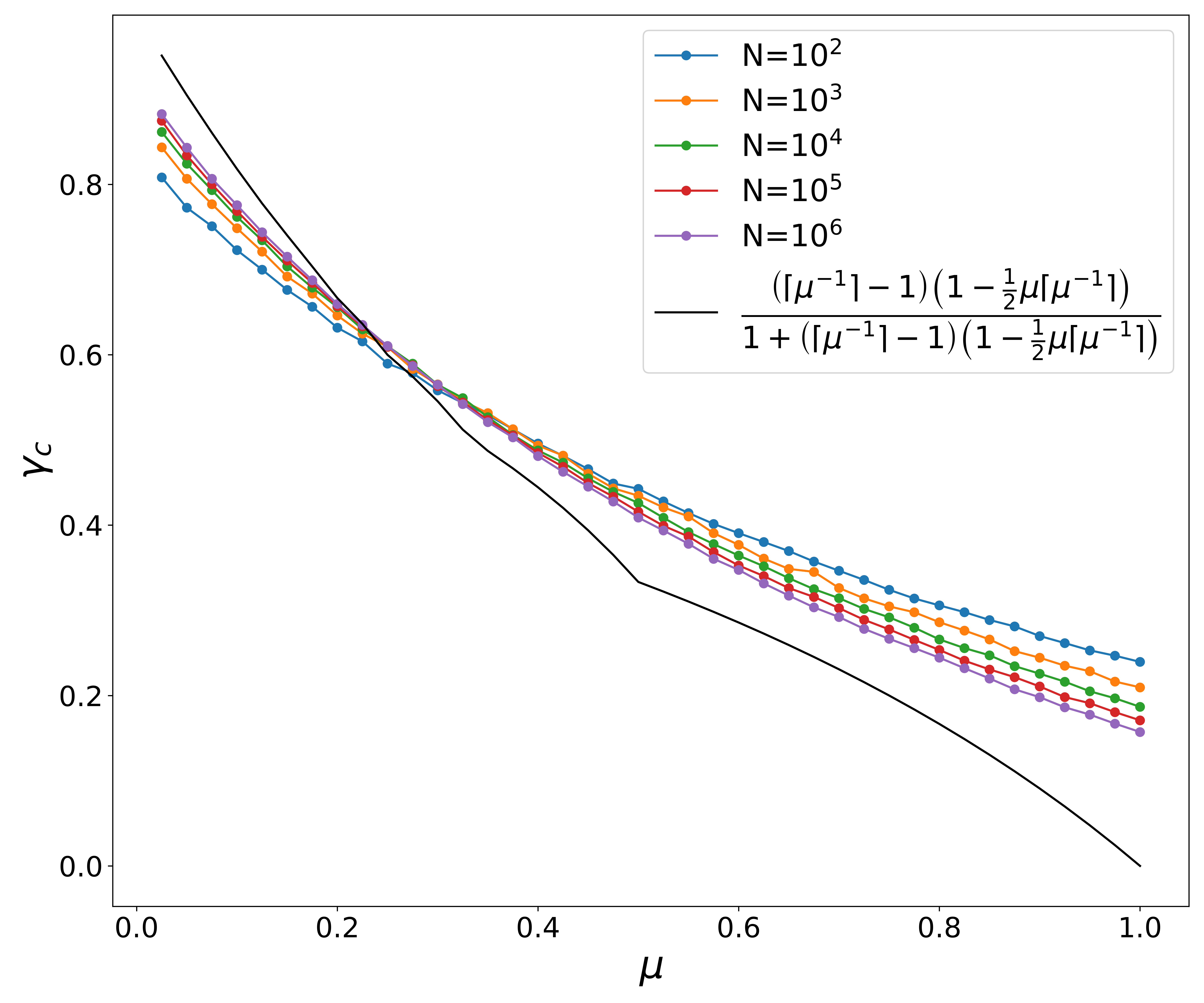}
    \end{tabular}
\end{center}
    \caption{
    {\it Left panel:}
     Rate of adaptation for the deterministic limiting model and its finite-size population correction, for the noisy exponential $K$-BRW.  
    {\it Right panel:} 
    We show the theoretical function $\mu\mapsto\gamma_c(\mu)$ given by Theorem~\ref{thm:main} in black, and, for different population sizes, the estimated critical values of $\gamma$, which maximize the function $\gamma \mapsto v^{(N)}_{\gamma,\mu}$ in the simulations.
  }
    \label{fig: rateOfCvSM}
\end{figure}

\section{Sexual reproduction}\label{sect: sexual}

We consider a sexual version of the asexual model exposed in  the main text. 
As before,  at every generation $n$, the population consists of $N^\gamma$ individuals. The population evolves from $n-1$ to $n$ into two successive steps.

\medskip

{\it Step 1. Reproduction.} We generate $N$ offspring. Each offspring has two parents $(x_1,x_2)$ chosen uniformly at random from generation $n-1$. The genotype and phenotype of an individual is
obtained by adding noise to the mean parental genotype values, that is $$\mbox{genotype}=\frac{1}{2}(x_1+x_2)+X, \ \ \ \mbox{phenotype}= \frac{1}{2}(x_1+x_2)+X+Y,$$
where $X$ and $Y$ have respective distribution $f_{X}$ and $f_{Y}$ with
$$
f_{X}(x) = \frac{1}{2} \exp(-|x| ),  \ \ f_{X}(x) = \frac{\mu}{2} \exp(-\mu |x| ),
$$
for a fixed value of $\mu>0$ which captures the inverse of the phenotypic noise.

{\it Step 2. Selection.} The population at generation $n$ is obtained by selecting the $N^\gamma$ genotypes with the $N^\gamma$ rightmost phenotypes; the same way as in the asexual case.

\medskip

{\it Log-profiles.}
Similarly to the asexual case, we can derive a recursive equation for the limiting stochastic exponents. 
At generation $n-1$, we think of 
$N^{g_{n-1}(x)} dx$
as the number of particles in $dx \log(N)$. We claim that if $g_0$ is concave on its support (that is the set of points where $g_{0}(x)\neq -\infty$), then for every $n\geq0$ 
\begin{equation}
  \label{eq:sex1}
  \begin{split}
      g_{n}(x) \ = \  \pi\left[ \sup_{y \in \R} \left( 1 - 2\gamma + 2 g_{n-1}(y) - |x-y| - \mu (s_n - x)_+\right)\right], \\
      \mbox{where $s_n$ satifies $\sup g_{n}= \gamma$}.
  \end{split}
\end{equation}

In order to justify the formula, we note that the average number of offspring with parents in $(dx_1,dx_2)$ is approximately
$$\frac{N^{g_{n-1}(x_1)}N^{g_{n-1}(x_2)}}{N^{2\gamma}} N dx_1 dx_2.$$
Along the lines of the asexual case, this yields
\begin{equation}
  \label{eq: sex2}
  \begin{split}
      g_{n}(x) \ = \  \pi\left[ \sup_{y_1,y_2 \in \R} \left( 1 - 2\gamma +  g_{n-1}(y_1)+ g_{n-1}(y_2) - \left|x-\frac{y_1+y_2}{2}\right| - \mu (s_n - x)_+\right)\right], \\
      \mbox{where $s_n$ satisfies $\sup g_{n}= \gamma$}.
  \end{split}
\end{equation}
Now, assume that $g_{n-1}$ is concave on its support. Then 
$$
\forall y_1,y_2 \in\mbox{Supp}(g_{n-1}), \ \   \ \
g_{n-1}(y_1)+g_{n-1}(y_2)  \ \leq \ 2 g_{n-1}\left(\frac{y_1+y_2}{2}\right),
$$
so that (\ref{eq:sex1})
must hold. Since $g_0$ is assumed to be concave on its support, it remains to show that concavity is preserved by the dynamics. This easily follows from Lemma \ref{lem:concave}.

It is interesting to observe the appearance of the term $2g_{n-1}(y)$ in~\eqref{eq:sex1} corresponding to the selection of two parents and that, in this deterministic limit, every individual surviving will have two parents with roughly equal genotypic values.

Even if the recursive equations of the sexual and asexual case look similar at first sight (see~\eqref{eq: deterministicDynamic2} in the main text), sexual populations exhibit a much richer behavior, and we postpone its mathematical analysis for future work.  

First, we see transitions similar to the asexual case when $\mu\in[1/2,1]$, that is, when the phenotypic noise is at intermediate values. In this case, Figure~\ref{Tux} shows that the speed of evolution is also maximized at intermediate values of $\gamma_{c}(\mu)$.  However, our simulations show a complex critical line with some resonance-like modes. In addition, we exhibit a regime which was not present in the asexual case. When the phenotypic noise is too high ($\mu<\frac{1}{2}$),  the speed of evolution stabilizes at $0$ indicating that the population remains static on the natural space scaling of the system (measured in $\log(N)$ units). See Figures \ref{Tux},\ref{Tux2}. 
Finally, our models predict that the speed of evolution is always higher in the asexual case.

\begin{figure}[H]
\centering
\begin{tabular}{cc}
  \includegraphics[width=0.4\textwidth]{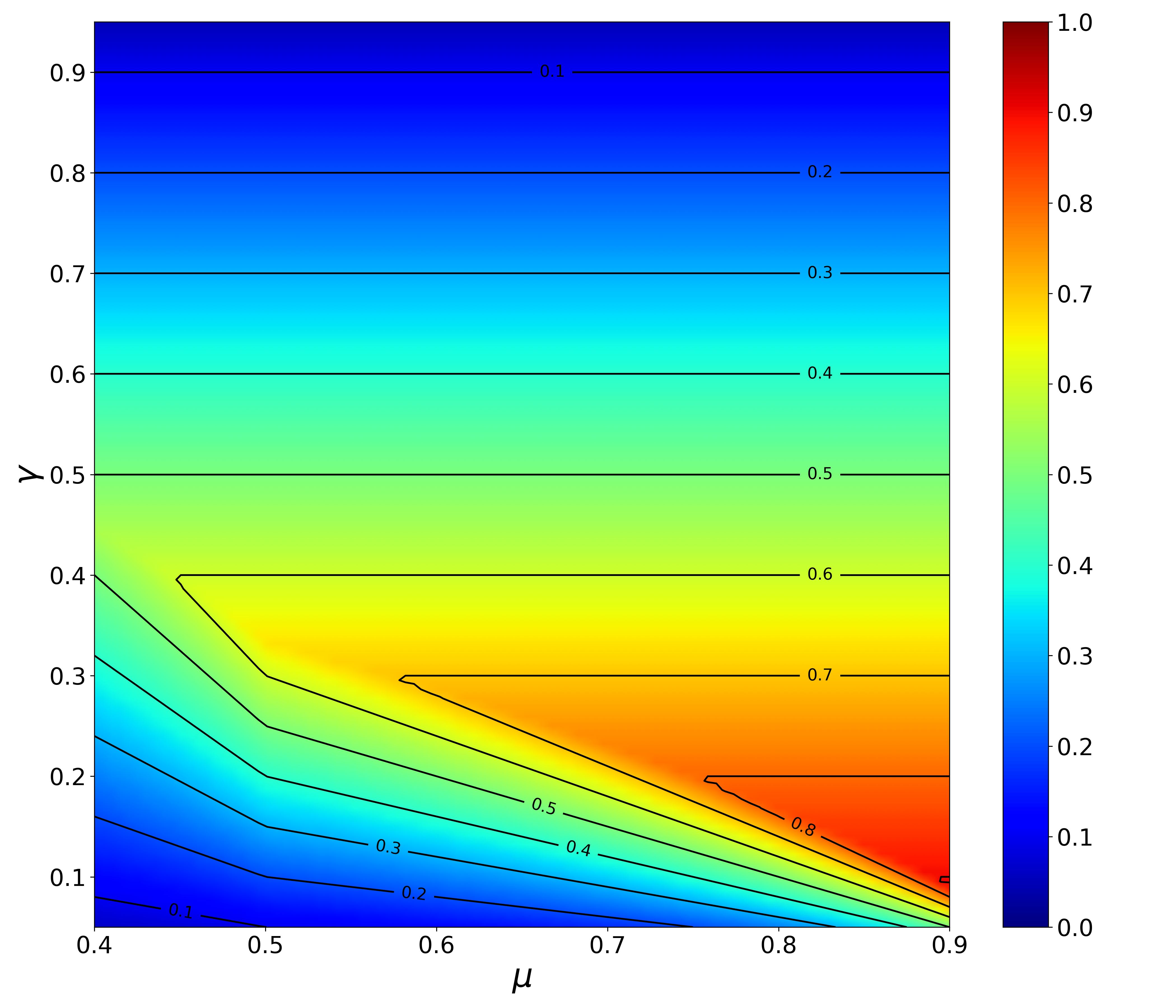}
& \includegraphics[width=0.4\textwidth]{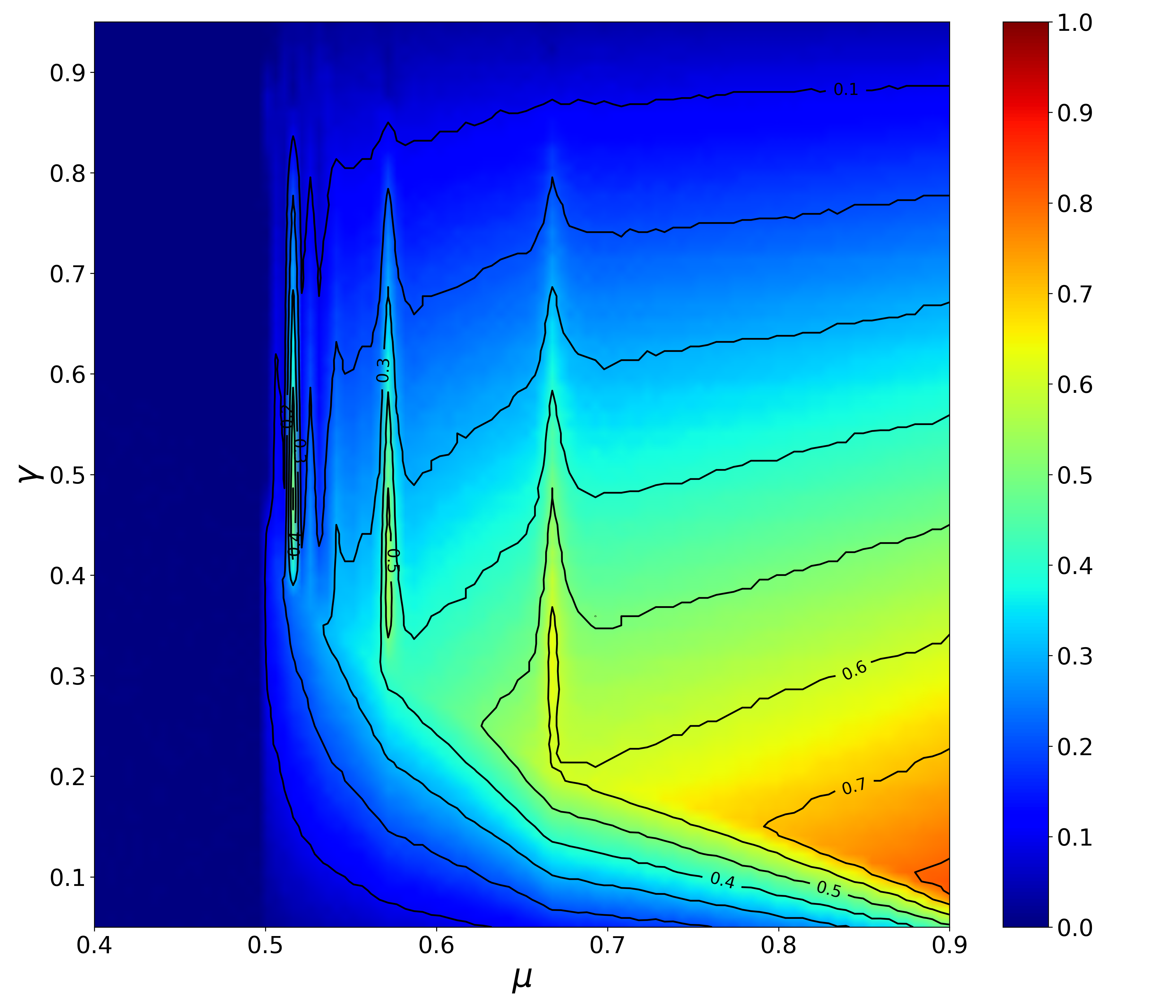}
\end{tabular}

 \includegraphics[width=0.4\textwidth]{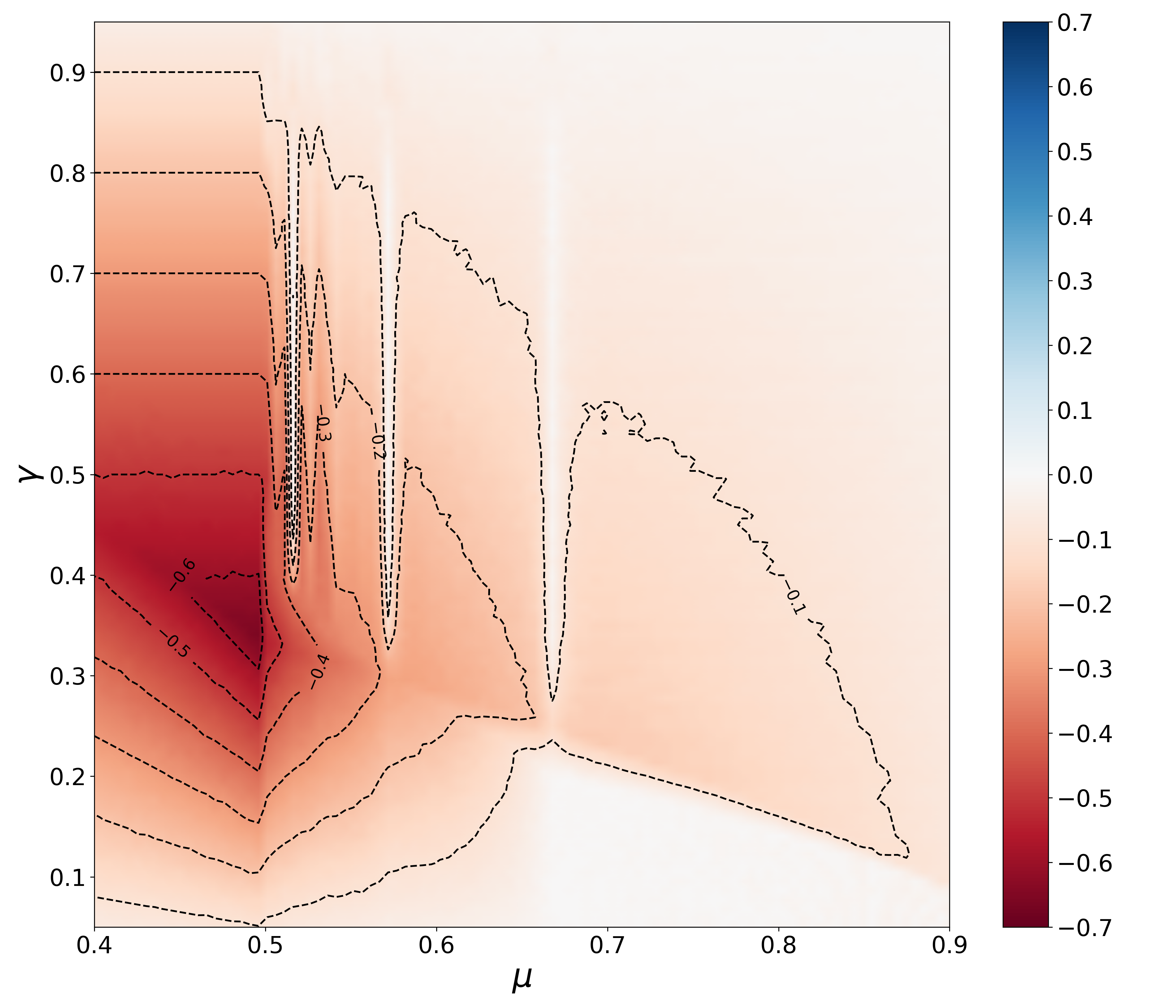}
\caption{{\it Left panel.} Speed for the asexual model as predicted by the explicit solution of (\ref{eq: deterministicDynamic2}) from the main text. {\it Right panel.} Speed for the sexual model as predicted by iterating the modified equation (\ref{eq:sex1}). For both models, if $\mu\in(.5,1)$, the speed is maximized at intermediate values. However, in the sexual case, the speed is always $0$ when the phenotypic noise is too large $\mu<1/2$. {\it Bottom panel.} Comparison of the speed of evolution in the sexual and asexual models.
}
\label{Tux}
\end{figure}

\begin{figure}[H]
\centering
\includegraphics[scale=.3]{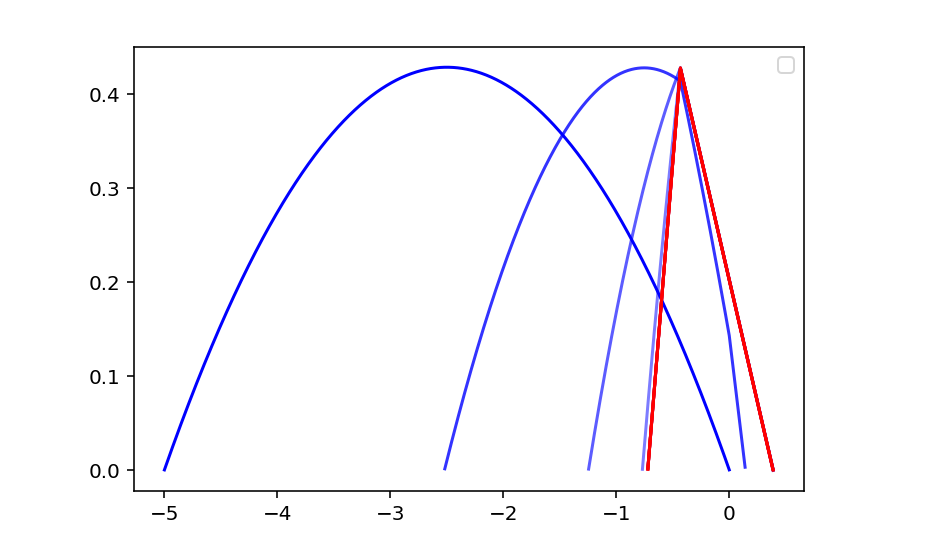}
 \ \ \ \ \includegraphics[scale=.3]{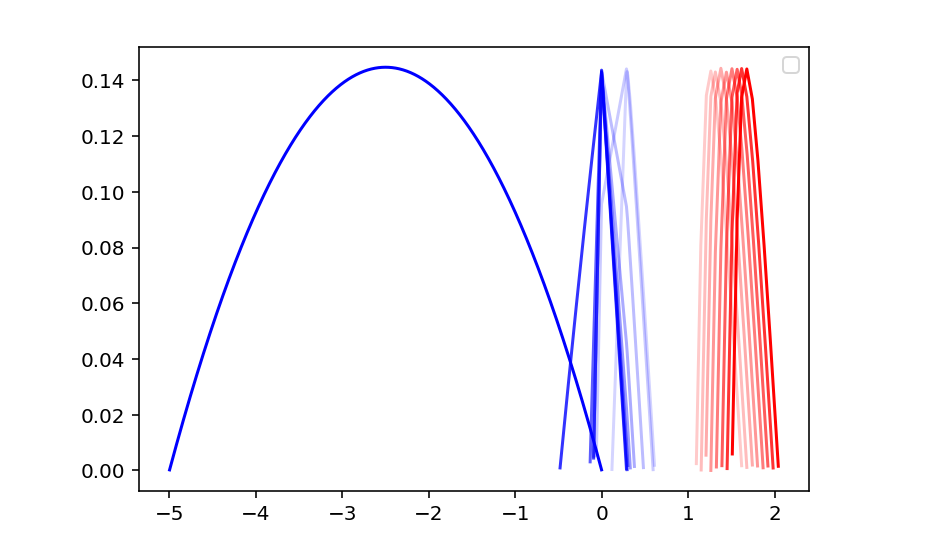}
 \ \ \  \
 \includegraphics[scale=.3]{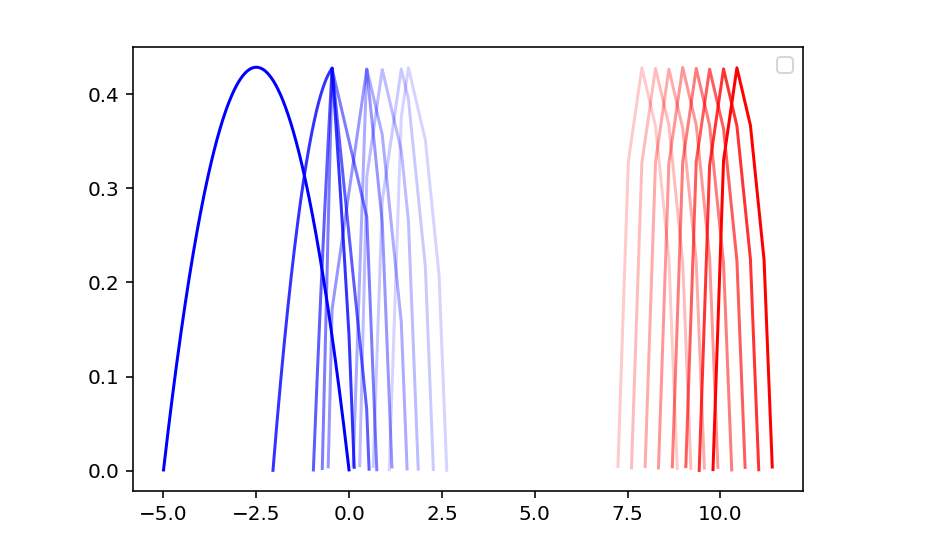}
\caption{Fitness wave for the sexual model. In blue : generation $0$ to $8$. In red: generation $16$ to $24$.  From left to right: static regime ($\mu=.48, \gamma= .43$); close to  static ($\mu=.51$, $\gamma=.14$); moving ($\mu=.61$, $\gamma=.61$)}
\label{Tux2}
\end{figure}

\end{document}